\let\accentvec\vec
\documentclass[smallextended]{svjour3} 

\let\vec\accentvec

\journalname{Data Mining and Knowledge Discovery}

\usepackage{amsmath}
\usepackage{amsfonts}
\usepackage{times}

\usepackage{natbib}
\usepackage{bbding} 
\usepackage{url}

\usepackage{booktabs}      
\usepackage{tikz}
\usetikzlibrary{snakes,arrows,shapes,positioning,fit,calc,matrix}
\usepackage{pgfplots}
\usepackage{ifpdf}
\ifpdf
\usepgfplotslibrary{external} 
\fi

\makeatletter

\newif\if@restonecol
\makeatother

\usepackage[ruled, vlined, linesnumbered, nofillcomment]{algorithm2e}
\usepackage{verbatim}

\smartqed

\usepackage[tight]{subfigure}
\usepackage[english]{babel}
\usepackage{textcomp}

\usepackage[pdftitle={Comparing Apples and Oranges - Measuring Differences between Data Mining Results},
           pdfauthor={Nikolaj Tatti and Jilles Vreeken},
		   hidelinks,
		   hyperfootnotes = false]{hyperref}

\newcommand{\set}[1]{\left\{#1\right\}}

\newcommand{\fpr}[1]{\mathopen{}\left(#1\right)}

\newcommand{\abs}[1]{{\left|#1\right|}}

\newcommand{\enset}[2]{\left\{#1 ,\ldots , #2\right\}}

\newcommand{\np}{\textbf{NP}}

\newcommand{\define}{\leftarrow}

\newcommand{\by}[2]{$#1$-by-$#2$}

\newcommand{\ifam}[1]{\mathcal{#1}}

\newcommand{\freq}[1]{\mathit{fr}\fpr{#1}}
\newcommand{\trans}[1]{\mathit{t}\fpr{#1}}
\newcommand{\attr}[1]{\mathit{a}\fpr{#1}}
\newcommand{\area}[1]{\mathit{s}\fpr{#1}}
\newcommand{\dist}[1]{\mathit{d}\fpr{#1}}

\newcommand{\dspace}{\mathcal{D}}

\newcommand{\ent}[1]{H\fpr{#1}}
\newcommand{\kl}[2]{\mathit{KL}\fpr{{#1} \,\|\, {#2}}}
\newcommand{\pemp}{p^*}

\newcommand{\efrac}[2]{\scriptscriptstyle\frac{#1}{#2}}

\SetKwComment{tcpas}{\{}{\}}
\SetCommentSty{textnormal}
\SetArgSty{textnormal}
\SetKw{False}{false}
\SetKw{True}{true}
\SetKw{Null}{null}
\SetKwInOut{Output}{output}
\SetKwInOut{Input}{input}
\SetKw{AND}{and}
\SetKw{OR}{or}
\SetKw{Continue}{continue}

\pgfdeclarelayer{background}
\pgfdeclarelayer{foreground}
\pgfsetlayers{background,main,foreground}

\definecolor{yafaxiscolor}{rgb}{0.3, 0.3, 0.3}

\definecolor{yafcolor1}{rgb}{0.4, 0.165, 0.553}
\definecolor{yafcolor2}{rgb}{0.949, 0.482, 0.216}
\definecolor{yafcolor3}{rgb}{0.47, 0.549, 0.306}
\definecolor{yafcolor4}{rgb}{0.925, 0.165, 0.224}
\definecolor{yafcolor5}{rgb}{0.141, 0.345, 0.643}
\definecolor{yafcolor6}{rgb}{0.965, 0.933, 0.267}
\definecolor{yafcolor7}{rgb}{0.627, 0.118, 0.165}
\definecolor{yafcolor8}{rgb}{0.878, 0.475, 0.686}

\newlength{\yafaxispad}
\newlength{\yaftlpad}
\newlength{\yaflabelpad}
\newlength{\yafaxiswidth}
\newlength{\yafticklen}
\makeatletter
\def\pgfplots@drawtickgridlines@INSTALLCLIP@onorientedsurf#1{}
\makeatother

\newcommand{\yafdrawxaxis}[2]{
	\pgfplotstransformcoordinatex{#1}\let\xmincoord=\pgfmathresult 
	\pgfplotstransformcoordinatex{#2}\let\xmaxcoord=\pgfmathresult 
	\pgfsetlinewidth{\yafaxiswidth} 
	\pgfsetcolor{yafaxiscolor}
	\pgfpathmoveto{\pgfpointadd{\pgfpointadd{\pgfplotspointrelaxisxy{0}{0}}{\pgfqpointxy{\xmincoord}{0}}}{\pgfqpoint{-0.5\yafaxiswidth}{\yafaxispad}}}
	\pgfpathlineto{\pgfpointadd{\pgfpointadd{\pgfplotspointrelaxisxy{0}{0}}{\pgfqpointxy{\xmaxcoord}{0}}}{\pgfqpoint{0.5\yafaxiswidth}{\yafaxispad}}}
	\pgfusepath{stroke}

}
\newcommand{\yafdrawyaxis}[2]{
	\pgfplotstransformcoordinatey{#1}\let\ymincoord=\pgfmathresult 
	\pgfplotstransformcoordinatey{#2}\let\ymaxcoord=\pgfmathresult 
	\pgfsetlinewidth{\yafaxiswidth} 
	\pgfsetcolor{yafaxiscolor}
	\pgfpathmoveto{\pgfpointadd{\pgfpointadd{\pgfplotspointrelaxisxy{0}{0}}{\pgfqpointxy{0}{\ymincoord}}}{\pgfqpoint{\yafaxispad}{-0.5\yafaxiswidth}}}
	\pgfpathlineto{\pgfpointadd{\pgfpointadd{\pgfplotspointrelaxisxy{0}{0}}{\pgfqpointxy{0}{\ymaxcoord}}}{\pgfqpoint{\yafaxispad}{0.5\yafaxiswidth}}}
	\pgfusepath{stroke}
}

\newcommand{\yafdrawaxis}[4]{\yafdrawxaxis{#1}{#2}\yafdrawyaxis{#3}{#4}}

\pgfplotscreateplotcyclelist{yaf}{%
{yafcolor1,mark options={scale=0.75},mark=o}, 
{yafcolor2,mark options={scale=0.75},mark=square},
{yafcolor3,mark options={scale=0.75},mark=triangle},
{yafcolor4,mark options={scale=0.75},mark=o},
{yafcolor5,mark options={scale=0.75},mark=o},
{yafcolor6,mark options={scale=0.75},mark=o},
{yafcolor7,mark options={scale=0.75},mark=o},
{yafcolor8,mark options={scale=0.75},mark=o}} 

\pgfplotsset{axis y line=left, axis x line=bottom,
	tick align=outside,
	compat = 1.3,
	tickwidth=\yafticklen,
	clip = false,
    x axis line style= {-, line width = 0pt, opacity = 0},
    y axis line style= {-, line width = 0pt, opacity = 0},
    x tick style= {line width = \yafaxiswidth, color=yafaxiscolor, yshift = \yafaxispad},
    y tick style= {line width = \yafaxiswidth, color=yafaxiscolor, xshift = \yafaxispad},
    x tick label style = {font=\scriptsize, yshift = \yaftlpad},
    y tick label style = {font=\scriptsize, xshift = \yaftlpad},
    every axis y label/.style = {at = {(ticklabel cs:0.5)}, rotate=90, anchor=center, font=\scriptsize, yshift = -\yaflabelpad},
    every axis x label/.style = {at = {(ticklabel cs:0.5)}, anchor=center, font=\scriptsize, yshift = \yaflabelpad},
    x tick label style = {font=\scriptsize, yshift = 1pt},
    grid = major,
    major grid style  = {dash pattern = on 1pt off 3 pt},
	every axis plot post/.append style= {line width=\yafaxiswidth} ,
	legend cell align = left,
	legend style = {inner sep = 1pt, cells = {font=\scriptsize}},
	legend image code/.code={%
		\draw[mark repeat=2,mark phase=2,#1] 
		plot coordinates { (0cm,0cm) (0.15cm,0cm) (0.3cm,0cm) };%
	} 
}

\xspaceaddexceptions{$}

\usepackage[labelfont=bf,labelsep=period]{caption}

\begin{document}

\title{Comparing Apples and Oranges
\thanks{The research described in this paper builds upon and extends the work appearing in \nobreak{ECML PKDD'11} as \citet{tatti:11:apples}.}
}
\subtitle{Measuring Differences between Exploratory Data Mining Results}

\author{Nikolaj Tatti \and Jilles Vreeken}
\institute{Nikolaj Tatti (\Envelope) \and Jilles Vreeken
 \at
	Advanced Database Research and Modelling, Department of Mathematics and Computer Science\\
	University of Antwerp, Antwerp, Belgium\\
	\email{nikolaj.tatti@ua.ac.be}
	\and
	Jilles Vreeken \at
	\email{jilles.vreeken@ua.ac.be}
}

\date{Received: date / Accepted: date}

\maketitle

\begin{abstract}

Deciding whether the results of two different mining algorithms provide significantly different information is an important, yet understudied, open problem in exploratory data mining. Whether the goal is to select the most informative result for analysis, or to decide which mining approach will most likely provide the most novel insight, it is essential that we can tell how \emph{different} the information is that \emph{different} results by possibly \emph{different} methods provide.

In this paper we take a first step towards comparing exploratory data mining results on binary data. We propose to meaningfully convert results into sets of noisy tiles, and  compare between these sets by Maximum Entropy modelling and Kullback-Leibler divergence, well-founded notions from Information Theory.
We so construct a measure that is highly flexible, and allows us to naturally include background knowledge, such that differences in results can be measured from the perspective of what a user already knows. Furthermore, adding to its interpretability, it coincides with Jaccard dissimilarity when we only consider exact tiles.

Our approach provides a means to study and tell differences between results of different exploratory data mining methods. As an application, we show that our measure can also be used to identify which parts of results best redescribe other results. Furthermore, we study its use for iterative data mining, where one iteratively wants to find that result that will provide maximal novel information.
Experimental evaluation shows our measure gives meaningful results, correctly identifies methods that are similar in nature, automatically provides sound redescriptions of results, and is highly applicable for iterative data mining.

\end{abstract}

\section{Introduction}\label{sec:intro}

Deciding whether the results of different mining algorithms provide significantly different information is an important, yet understudied, open problem in exploratory data mining. Whether we want to select the most promising result for analysis by an expert, or decide which mining approach we should apply next in order to most likely gain novel insight, we need to be able to tell how different the information is that is provided by different results, possibly from different methods. However, while the comparison of results is a well-studied topic in statistics, it has received much less attention in the knowledge discovery community.

Clearly, any dataset only contains a limited amount of knowledge---which is the most that we can hope to discover from it. To extract this information, we have at our disposal an ever growing number of data mining algorithms. 
However, when analysing data we have to keep in mind that most data mining results are complex, and their analysis and validation hence often requires considerable effort and/or cost. 
So, simply applying `all' methods and having an expert analyse `all' results is not a feasible approach to extract `all' knowledge.
Moreover, many of these results will be redundant, i.e. they will convey roughly the same information, and hence analysing `all' obtained results will mostly require effort that will not provide extra insight.

Instead, we would ideally just select that result for analysis that will provide us the most new knowledge. In order to be able to do this, two basic requirements have to be met. First of all, we need to be able to measure how different two results are from an information-providing perspective; if they essentially provide the same information, we could just select one for processing. Second, we should be able to include our background knowledge, such that we can gauge the amount of information included in a result compared to what we already know. 

Although an important practical problem, it has been surprisingly understudied in data mining.
The main focus in exploratory data mining research has mostly been on developing techniques to discover structure, and not so much on how to compare between results of different methods.
As a result there currently exist no general methods or theory to this end in the data mining literature. 

For tasks where a formal objective is available, it makes sense to straightforwardly use that objective for comparing fairly between the results of different methods. In classification, for instance, we can use accuracy, or AUC scores---as well as that we can measure how differently classifiers behave, for instance in ensembles~\citep{kuncheva:03:classdiff}. In fact, in learning tasks in general, there typically are clearly defined goals, and hence different methods can be compared on their performance on that particular task.

For exploratory data mining, however, there is no formal common goal: the main goal is the \emph{exploration} of the data at hand, and hence basically poses the rather general question `what can you tell me about my data?'. In fact, any type of data mining result that provides the analyst novel insight in the data, and the process that generated it, is potentially useful; regardless of whether the provided result of that method is, for example, a clustering, some subspace clusters, a collection of itemsets or correlated patterns, a decision tree, or even the identification of some statistical properties; as long as the analyst can obtain insight from that result, any data mining method can provide a valid (partial) answer to this question. Clearly, however, these methods do not share a common formal goal, and hence their results that are not easily comparable.
The core of the problem is that comparing between methods is thus like comparing \emph{apples} to \emph{oranges}: a clustering is a different result than a set of itemsets, which are, in turn, different from a classifier, set of subgroups, etc. So, in order to make a sensible comparison, we need to find a common language.

In this regard, the comparison of complex objects, e.g. of datasets~\citep{tatti:07:distances,vreeken:07:difference}, is related. Our setting, however, is more general, as now we do not want to compare between one type of complex object, but want to consider a very rich class of objects---a class potentially consisting of \textit{any} data mining result.
Arguably, some of the most general complex objects to compare between are probability distributions. Statistics and Information Theory provide us tools for measuring differences between distributions, such as Kullback-Leibler divergence~\citep{cover:06:elements}. Data mining results, however, rarely are probability distributions, and if they are, not necessarily for the same random variable; making these tools unsuited for direct application for our goal. 

A simple yet important observation is that any mining result essentially identifies some properties of the dataset at hand. In fact, in an abstract way, we can identify all datasets for which these properties hold. 
This is an important notion, as it provides us with a way to compare between results of different methods: if two results provide the same information, they identify the same subspace of possible datasets, and the more different the information two results give, the smaller the overlap between the sets of possible datasets will be. 

As such, a straightforward approach for measuring similarity between results would be to simply count the number of possible datasets per result, and regard the relative number in the overlap---the larger this number, the more alike the information provided. However, the most basic situation of rectangles full of  ones aside, counting databases is far from trivial. Instead of counting, we can view our problem more generally, and regard a result as implicitly defining a probability distribution over datasets. And hence, if we can model these distributions, we can use standard tools from Statistics to compare between data mining results fair and square. 

In this paper, we propose to translate data mining results into probability distributions over datasets using the Maximum Entropy principle~\citep{csiszar:75:i-divergence}. This allows us to uniquely identify that probabilistic model that makes optimal use of the information provided by a result, but is least biased otherwise. By subsequently measuring the Kullback-Leibler divergence between the obtained distributions, we can tell how different two results are from an information perspective. 
Furthermore, besides directly measuring the information content between data mining results, we can also incorporate background knowledge into these models, and so score and compare informativeness of results from specific points of view.

Finding these probability distributions, and conditioning them using data mining results, however, is far from trivial in general. Here, we give a proof of concept of this approach for binary data, for which the basics of maximum entropy modelling are available~\citep{debie:11:dami}.
We show that many types of results of exploratory data mining approaches on binary data can easily and meaningfully be translated into sets of noisy tiles: combinations of rows and columns, for which we know the density of $1$s. We show we can efficiently acquire the maximum entropy distribution over datasets given such sets of noisy tiles, and construct a flexible measure that uses KL divergence by which we can compare between results.

To show that our measure works in practice, we compare between the results of ten different exploratory data mining methods, including (bi-)clusters, subspace clusters, sets of tiles, and sets of (frequent) itemsets. We further give a theoretic framework to mine for redescriptions. That is, given a (sub)set of noisy tiles from one result, we can identify the set of tiles from another result that best approximates the same information. 
Moreover, we show that our measure is highly effective for iterative data mining: given a (large) collection of results, we can efficiently discover which result provides the most novel information, and which hence is a likely candidate to be analysed next.
Experiments show our measure works well in practice: dissimilarity converges to $0$ when results approximate each other, methodologically close methods are correctly grouped together, sensible redescriptions for tile-sets are obtained, and meaningful orders of analysis are provided. 
In other words, we give an approach by which we can meaningfully \emph{mix} {apples} {and} {oranges}, and compare between them fairly.

In this paper we build upon and extend the work published as~\citep{tatti:11:apples}. Here, we discuss the
theory, properties, and choices of our measure in closer detail, as well as the limitations of this proof-of-concept measure. Most importantly, we investigate the practical application of our measure for both mining redescriptions of (partial) results, and for application to the end of iterative data mining---giving algorithms for both these settings, and providing experimental validation that they provide meaningful results. 

The road map of this paper is as follows. Next, in Section~\ref{sec:prelim} we give the notation and preliminaries we use throughout the paper, and discuss how to convert results on binary data into sets of tiles in Section~\ref{sec:tiles}. Section~\ref{sec:model} details how we can build a global model from a set of tiles, which we use in Section~\ref{sec:compare} to define a measure to compare such sets. In Section~\ref{sec:appl} we subsequently use this measure for redescribing sets of tiles, as well as to iteratively identify the most informative result. We discuss related work in Section~\ref{sec:related}, and we evaluate our measure empirically in Section~\ref{sec:exps}. We round up with discussion and conclusions in Sections~\ref{sec:disc} and \ref{sec:concl}. We give the proofs, such as for NP-completeness, in Appendix~\ref{sec:apx}.

\section{Preliminaries}\label{sec:prelim}

In this section, we define the preliminaries we will use in subsequent sections.

A \emph{binary dataset} $D$ is a binary matrix of size \by{n}{m} consisting
of $n$ rows, or transactions, and $m$ columns, or attributes. A row is simply a binary vector of size $m$. We assume that both the rows and columns of $D$ have unique integer identifiers such that we can easily refer to them. 
We denote the $(i, j)^{\textrm{th}}$ entry of $D$ by $D(i, j)$. 
We denote the space of all \by{n}{m} binary datasets by $\dspace$.

Given two distributions, $p$ and $q$, defined over $\dspace$, we define entropy as
\[
\ent{p} = -\sum_{D \in \dspace} p(D) \log p(D)
\]
and Kullback-Leibler divergence as
\[
\kl{p}{q} = \sum_{D \in \dspace} p(D) \log \frac{p(D)}{q(D)}.
\]

All logarithms are natural logarithms, to base $e$. We employ the usual convention of $0\log 0 = 0$.

\section{Results on Binary Data as Noisy Tiles}\label{sec:tiles}

As discussed in the introduction, in this paper we give a proof of concept for measuring differences between data mining results by first converting these results into maximum entropy distributions over datasets, and subsequently measuring the difference in information content between these distributions. 
Slightly more formally, the first step entails inferring a maximum entropy distribution for the provided background knowledge. The background knowledge, in this case, being the data mining result at hand. To do so, we need theory on how to do this for the particular type of background knowledge---which is far from trivial in general, and hence there exists no theory for converting data mining results in general into (maximum entropy) probability distributions. 

Recently, however, \citet{debie:11:dami} formalised the theory for inferring the maximum entropy distribution for a binary dataset given so-called \emph{noisy tiles} as background knowledge. We employ this result for our proof of concept, and hence focus on measuring differences between data mining results on binary data. While certainly not all data in the world is binary, surprisingly many (exploratory) data mining methods are aimed at extracting significant structure from 0--1 data. 

First of all, there are approaches that identify whether data exhibits particular structure. As an example for binary data, the most simple example is density: what is the relative number of $1$s of the data. A more intricate type of structure we can identify in this type of data is bandedness~\citep{garriga:11:banded}: the property that the rows and columns of the data can be permutated such that the non-zero entries exhibit a staircase pattern of overlapping rows. Nestedness~\citep{mannila:07:nestedness}, is the property that for every pair of rows of a $0$--$1$ dataset one row is either a superset or subset of the other.

Perhaps the most well-known example for analysing binary data is frequent itemset mining~\citep{agrawal:94:fast}. Given appropriate distance functions, binary data can be clustered~\citep{macqueen:67:some}, or subspace clusters (clusters within specific subsets of dimensions)~\citep{aggarwal:99:proclus,muller:09:sscl} can be identified. Both traditional pattern mining and sub-space clustering typically provide incredibly many results. In response, a recent development is pattern \emph{set} mining, where the goal is not to find all patterns (e.g. itemsets) that satisfy some conditions (e.g. frequency), but to instead find small, non-redundant, high-quality \emph{sets} of patterns that together generalise the data well~\citep{siebes:06:item,miettinen:08:discrete,knobbe:06:pattern}. An early proponent of this approach is Tiling~\citep{geerts:04:tiling}, which has the goal of finding large tiles, i.e. itemsets that cover many $1$s in the data.  \citet{konto:10:sdm} proposed to mine for \emph{noisy} tiles that are surprising given a background model of the data. We discuss mining sets of patterns or tiles, as well as mining binary data in general, in more detail in Section~\ref{sec:related}.

Interestingly, as we will show below, many data mining results on binary data, including clustering and pattern mining results, can be translated into noisy tiles without much, or any loss of information. Before we discuss how different results can be translated, we have to formally introduce the concept of tiles.

A \emph{tile} $T = (\trans{T}, \attr{T})$ is a tuple consisting of two lists. The
first list, $\trans{T}$, is a set of integers between $1$ and $n$ representing
the transaction ids of $T$. The second list, $\attr{T}$, is a set of integers
between $1$ and $m$ representing the attribute ids. We define $\area{T}$ to give us all entries of $T$, that is, the Cartesian product of $\trans{T}$ and $\attr{T}$, $\area{T} = \set{(i, j) \mid i \in \trans{T}, j \in \attr{T}}$. The number of entries of $\area{T}$, $\abs{\area{T}}$, we refer to as the \emph{area} of a tile $T$.
Given a tile set $\ifam{T}$ we also define $\area{\ifam{T}} = \bigcup_{T \in \ifam{T}} \area{T}$ to give us the list of all entries $\ifam{T}$ covers.

Given a tile $T$ and a dataset $D$ we define a frequency $\freq{T ; D}$ to
be the proportion of ones in $D$ corresponding to the entries identified by $T$,
\[
	\freq{T; D} = \frac{1}{\abs{\area{T}}} \sum_{i \in \trans{T}} \sum_{j \in \attr{T}} D(i, j) \quad .
\]

\begin{example}
Assume that we are given a set of text documents, for example, a set of
abstracts of accepted papers at a computer science conference. We can transform
these documents into a binary dataset using a bag-of-words representation. That is, each item represents a particular word, and every document is represented by a binary vector $t$ with entries $t_i = 1$ if and only if the corresponding word $i$ occurs in the abstract. Now, let $T$ be a tile in this dataset. Then the
transactions of $T$ correspond to specific abstracts, while the columns, or the items of $T$ represent specific words. The frequency of $1$s in $T$ is simply the relative number of these words occurring in these documents.  For example, when we run the \textsc{Krimp} algorithm on such a dataset (see
Section~\ref{sec:exps} for more details on this algorithm), we see it finds (among others) tiles containing
the words \emph{$\{$large, database$\}$}, \emph{$\{$algorithm, result, set, experiment$\}$}, and \emph{$\{$synthetic, real$\}$}. In this case, for each of the discovered tiles the corresponding frequencies are $1.0$, and hence we know that each of the documents identified by these tiles contains all of the  corresponding words.
\end{example}

Let $p$ be a distribution defined over $\dspace$, the space of all \by{n}{m} datasets.
We define the frequency of a tile to be the average frequency with respect to $p$,
\[
	\freq{T; p} = \sum_{D \in \dspace} p(D) \freq{T ; D} \quad .
\]
We can also express the frequency directly by this distribution.
\begin{lemma}
\label{lem:alternative}
Given a distribution $p$ and a tile $T$, the frequency is equal to
\[
	\freq{T; p} = \frac{1}{\abs{\area{T}}} \sum_{(i, j) \in \area{T}} p((i, j) = 1),
\]
where $p((i, j) = 1)$ is the probability of a dataset having  $1$ as $(i, j)$th entry.
\end{lemma}

We say a tile is \emph{exact} if its frequency is $0$ or $1$, and otherwise say  it is \emph{noisy}.

\begin{corollary}
\label{cor:exact}
For an exact tile $T$, $p((i, j) = 1) = \freq{T; p}$, where $(i, j) \in \area{T}$.
\end{corollary}

\begin{example}
\label{ex:toy1}
Consider a dataset $D$ given in Figure~\ref{fig:toy:data}. We consider five different tiles,
\begin{itemize}
\item $T_1 = (2, \ldots, 5) \times (1, \ldots, 5)$ of area $20$ and covering $10$ ones of $D$, 

\item 
$T_2 = (1, 2) \times (1, 2)$ of area $4$ and covering $4$ ones,

\item 
$T_3 = (3, 4, 5) \times (1, 2)$ of area $6$ and covering $0$ ones, 

\item 
$T_4 = (4, 5) \times (3, 4, 5)$ of area $6$ and covering $6$ ones, 

\item 
and $T_5 = (3, 4, 5) \times (4, 5)$ of area $6$ and also covering $6$ ones.
\end{itemize}
Their subsequent frequencies are hence $\freq{T_1; D} = 10/20 = 1/2$ for $T_1$, while $\freq{T_2; D} = \freq{T_4; D} = \freq{T_5; D} = 1$, and $\freq{T_3; D} = 0$.
By definition, $T_1$ is a noisy tile, while $T_2, \ldots, T_5$ are exact tiles.
\end{example}

There are numerous techniques for mining tiles and sets from a database, but
we observe that a large number of other statistics and mining results obtained on binary data can also be naturally described using sets of (noisy) tiles: 
\begin{itemize}
\item \emph{itemsets}: any itemset can be converted into a tile by taking the supporting transactions. For standard frequent itemsets this results in an exact tile, for fault-tolerant (or noisy) itemsets this results in a noisy tile.
Thus, an itemset collection can be converted into a tile set.

\item \emph{clustering}: Given a clustering, we can construct a tile set in two
different ways. The first way is to represent each cluster by a single tile,
consisting of the transactions in the cluster and all the items.  The other way
is to transform each cluster into $m$ tiles, where $m$ is the number of items
in the data. Each tile corresponds to an item and the transactions associated
with the cluster. The frequency corresponding to a tile then simply is the relative occurrence of that item in that cluster, i.e. its mean value. This is particularly natural for $k$-means, since a
centroid then corresponds to the column means of the corresponding
transactions.

\item \emph{bi-clustering \emph{and} subspace-clustering}: both subspace clusters~\citep{aggarwal:99:proclus} and bi-clusters~\citep{pensa:05:bicluster} are sets of transactions and columns. Hence, we can naturally represent these results by equivalent tiles.

\item \emph{data density}: the density of the data is equal to the frequency of a tile containing the whole data.

\item \emph{column and row margins}: the margin of a column $i$, i.e. its frequency, can be expressed with a single (noisy) tile containing the column $i$
and all rows $r \in D$. Analogously, we can express row margins by creating a tile per row over all columns.
\end{itemize}

As such, we can convert a wide range of results on binary data into sets of (noisy) tiles---and retain most, if not all, information provided by these results. 
Note, however, that this list is far from complete: we cannot convert every result on binary data into (sets of) (noisy) tiles. For instance, we can currently not model structures \emph{within} a tile beyond its density, and hence in this proof-of-concept we disregard
statistics that go beyond density, such as Lazarus counts~\citep{fortelius:06:spectral}, or nestedness~\citep{mannila:07:nestedness}. 
We discuss converting data mining results into tiles, and the limitations of the current proof-of-concept, in more detail in Section~\ref{sec:disc}.

\section{Building Global Models from Tiles}\label{sec:model}

To meet our goal, we have to construct a statistically sound technique for comparing two sets of tiles. In this section we construct a global model for datasets using the given tiles. We will use these models for comparing the tile sets.

Consider that we are given a tile set $\ifam{T}$, and for each tile $T \in
\ifam{T}$ we are also given a frequency $\alpha_T$. Typically, the frequencies
are obtained from the data at hand, $\alpha_T = \freq{T ; D_{\mathit{in}}}$, but this is not a necessary
condition. The tiles convey local information about the data $D_{in}$ and our goal
is to infer a distribution $p$ over $\dspace$, that is, how probable data set $D \in \dspace$ is given a tile set $\ifam{T}$. If the information at hand defines the data set uniquely,
 then $p(D) = 1$ if and only if $D = D_{in}$.

To derive the model, we use a well-founded notion from information theory, the Maximum Entropy principle~\citep{csiszar:75:i-divergence,jaynes:82:rationale}.
Roughly speaking, by Maximum Entropy, we incorporate the given information into a distribution, yet further making it as evenly spread as
possible. To define the distribution, we first define the space of distribution candidates. That is, the space of those distributions that produce the same frequencies for the given tiles,
	$\mathcal{P} = \set{p \mid \freq{T ; p} = \alpha_T, \text{ for all } T \in \ifam{T}}$.
In other words, $\mathcal{P}$ contains all distributions that explain the
frequencies $\alpha_T$. From this set, we select one distribution, which we
denote by $\pemp_\ifam{T}$, such that $\pemp_\ifam{T}$ maximises the entropy,
$\ent{\pemp_\ifam{T}} \geq \ent{p}$ for any $p \in \mathcal{P}$.

We will abuse notation and write $\ent{\ifam{T}}$ where we mean $\ent{\pemp_\ifam{T}}$.
Similarly we write $\kl{\ifam{T}}{\ifam{U}}$ to mean $\kl{\pemp_\ifam{T}}{\pemp_\ifam{U}}$,
where both $\ifam{T}$ and $\ifam{U}$ are tile sets.

\tikzstyle{cell} = [inner sep = 1pt]
\tikzstyle{tile} = [rounded corners = 2pt, inner sep = 0pt, fill opacity = 0.3]

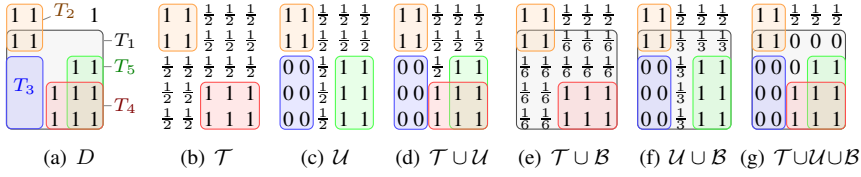
\begin{figure}[tb!]
\centering
\subfigure[$D$\label{fig:toy:data}]{
\begin{tikzpicture} 
\tikzset{every pin/.style={rectangle,font=\scriptsize, color = black, pin distance = 1pt, pin edge={shorten >=-2pt}}}

\matrix[ampersand replacement = \&] {
	\node [cell] (t1n1) {1} ; \&
	\node [cell] {1} ; \& \&
	\node [cell] {$\phantom{\efrac{1}{2}}$}; \&
	\node [cell] {1} ;
	\\
	\node [cell] (t5n1) {1} ; \&
	\node [cell]  (t1n2) {1} ; \&
	\node [cell] {$\phantom{\efrac{1}{2}}$}; \&
	\&
	\node [cell] (t1n3) {$\phantom{1}$}; \&
	\\
	\node [cell] (t2n1) {$\phantom{0}$}; \& \&
	\node [cell] {$\phantom{\efrac{1}{2}}$}; \&
	\node [cell] (t3n1) {1} ; \&
	\node [cell] (t5n3) {1} ;
	\\
	\&
	\node [cell] {$\phantom{\efrac{1}{2}}$}; \&
	\node [cell] (t4n1) {1} ; \&
	\node [cell] {1} ; \&
	\node [cell] {1} ;
	\\
	\node [cell] (anc) {$\phantom{0}$}; \&
	\node [cell] (t2n2) {$\phantom{0}$}; \&
	\node [cell] {1} ; \&
	\node [cell] {1} ; \&
	\node [cell] (t34n2) {1} ; \&
	\node [cell] {$\vphantom{\efrac{1}{2}}$}; \\
};
\begin{pgfonlayer}{background}
\pgfsetstrokeopacity{1} 
\node[tile, fit=(t5n1) (t34n2), draw=black!70, fill=black!10] {};
\node[tile, fit=(t5n1) (t1n3), pin=0:{$T_1$}] {};
\node[tile, fit=(t1n1) (t1n2), draw=orange!70, fill=orange!30, pin={[orange!50!black]5:$T_2$}]  {};
\node[tile, fit=(t2n1) (t2n2), draw=blue!70, fill=blue!30] {};
\node[fit=(t2n1) (t2n2), font=\scriptsize, text=blue]  {$T_3$};
\node[tile, fit=(t3n1) (t34n2), draw=green!70, fill=green!30]  {};
\node[tile, fit=(t2n1) (t5n3), pin={[green!50!black]0:$T_5$}]  {};
\node[tile, fit=(t4n1) (t34n2), draw=red!70, fill=red!30, pin={[red!50!black]0:$T_4$}]  {};

\end{pgfonlayer}
\end{tikzpicture}}%
\subfigure[$\ifam{T}$\label{fig:toy:model1}]{
\begin{tikzpicture} 
\matrix[ampersand replacement = \&] {
	\node [cell] (t1n1) {1} ; \&
	\node [cell] {1} ; \&
	\node [cell] {$\efrac{1}{2}$} ; \&
	\node [cell] {$\efrac{1}{2}$} ; \&
	\node [cell] {$\efrac{1}{2}$} ; \&
	\\
	\node [cell] (t5n1) {1} ; \&
	\node [cell] (t1n2) {1} ; \&
	\node [cell] {$\efrac{1}{2}$} ; \&
	\node [cell] {$\efrac{1}{2}$} ; \&
	\node [cell] {$\efrac{1}{2}$} ; \&
	\\
	\node [cell] {$\efrac{1}{2}$} ; \&
	\node [cell] {$\efrac{1}{2}$} ; \&
	\node [cell] {$\efrac{1}{2}$} ; \&
	\node [cell] {$\efrac{1}{2}$} ; \&
	\node [cell] {$\efrac{1}{2}$} ; \&
	\\
	\node [cell] {$\efrac{1}{2}$} ; \&
	\node [cell] {$\efrac{1}{2}$} ; \&
	\node [cell] (t4n1) {1} ; \&
	\node [cell] {1} ; \&
	\node [cell] {1} ; \&
	\\
	\node [cell] (anc) {$\efrac{1}{2}$} ; \&
	\node [cell] {$\efrac{1}{2}$} ; \&
	\node [cell] {1} ; \&
	\node [cell] {1} ; \&
	\node [cell] (t34n2) {1} ; \&
	\\
};
\begin{pgfonlayer}{background}
\node[tile, fit=(t1n1) (t1n2), draw=orange!70, fill=orange!30]  {};
\node[tile, fit=(t4n1) (t34n2), draw=red!70, fill=red!30]  {};
\end{pgfonlayer}
\end{tikzpicture}}%
\subfigure[$\ifam{U}$]{
\begin{tikzpicture} 
\matrix[ampersand replacement = \&] {
	\node [cell] (t1n1) {1} ; \&
	\node [cell] {1} ; \&
	\node [cell] {$\efrac{1}{2}$} ; \&
	\node [cell] {$\efrac{1}{2}$} ; \&
	\node [cell] {$\efrac{1}{2}$} ; \&
	\\
	\node [cell] (t5n1) {1} ; \&
	\node [cell] (t1n2) {1} ; \&
	\node [cell] {$\efrac{1}{2}$} ; \&
	\node [cell] {$\efrac{1}{2}$} ; \&
	\node [cell] {$\efrac{1}{2}$} ; \&
	\\
	\node [cell] (t2n1) {0} ; \&
	\node [cell] {0} ; \&
	\node [cell] {$\efrac{1}{2}$} ; \&
	\node [cell] (t3n1) {1} ; \&
	\node [cell] {1} ; \&
	\\
	\node [cell] {0} ; \&
	\node [cell] {0} ; \&
	\node [cell] {$\efrac{1}{2}$} ; \&
	\node [cell] {1} ; \&
	\node [cell] {1} ; \&
	\\
	\node [cell] (anc) {0} ; \&
	\node [cell] (t2n2) {0} ; \&
	\node [cell] {$\efrac{1}{2}$} ; \&
	\node [cell] {1} ; \&
	\node [cell] (t34n2) {1} ; \&
	\\
};
\begin{pgfonlayer}{background}
\node[tile, fit=(t1n1) (t1n2), draw=orange!70, fill=orange!30]  {};
\node[tile, fit=(t2n1) (t2n2), draw=blue!70, fill=blue!30]  {};
\node[tile, fit=(t3n1) (t34n2), draw=green!70, fill=green!30]  {};
\end{pgfonlayer}
\end{tikzpicture}}%
\subfigure[$\ifam{T} \cup \ifam{U}$]{
\begin{tikzpicture} 
\matrix[ampersand replacement = \&] {
	\node [cell] (t1n1) {1} ; \&
	\node [cell] {1} ; \&
	\node [cell] {$\efrac{1}{2}$} ; \&
	\node [cell] {$\efrac{1}{2}$} ; \&
	\node [cell] {$\efrac{1}{2}$} ; \&
	\\
	\node [cell] (t5n1) {1} ; \&
	\node [cell] (t1n2) {1} ; \&
	\node [cell] {$\efrac{1}{2}$} ; \&
	\node [cell] {$\efrac{1}{2}$} ; \&
	\node [cell] {$\efrac{1}{2}$} ; \&
	\\
	\node [cell] (t2n1) {0} ; \&
	\node [cell] {0} ; \&
	\node [cell] {$\efrac{1}{2}$} ; \&
	\node [cell] (t3n1) {1} ; \&
	\node [cell] {1} ; \&
	\\
	\node [cell] {0} ; \&
	\node [cell] {0} ; \&
	\node [cell] (t4n1) {1} ; \&
	\node [cell] {1} ; \&
	\node [cell] {1} ; \&
	\node [cell] {$\vphantom{\efrac{1}{2}}$}; 
	\\
	\node [cell] (anc) {0} ; \&
	\node [cell] (t2n2) {0} ; \&
	\node [cell] {1} ; \&
	\node [cell] {1} ; \&
	\node [cell] (t34n2) {1} ; \&
	\node [cell] {$\vphantom{\efrac{1}{2}}$}; 
	\\
};
\begin{pgfonlayer}{background}
\node[tile, fit=(t1n1) (t1n2), draw=orange!70, fill=orange!30]  {};
\node[tile, fit=(t2n1) (t2n2), draw=blue!70, fill=blue!30]  {};
\node[tile, fit=(t3n1) (t34n2), draw=green!70, fill=green!30]  {};
\node[tile, fit=(t4n1) (t34n2), draw=red!70, fill=red!30]  {};
\end{pgfonlayer}
\end{tikzpicture}}%
\subfigure[$\ifam{T} \cup \ifam{B}$]{
\begin{tikzpicture} 
\matrix[ampersand replacement = \&] {
	\node [cell] (t1n1) {1} ; \&
	\node [cell] {1} ; \&
	\node [cell] {$\efrac{1}{2}$} ; \&
	\node [cell] {$\efrac{1}{2}$} ; \&
	\node [cell] {$\efrac{1}{2}$} ; \&
	\\
	\node [cell] (t5n1) {1} ; \&
	\node [cell] (t1n2) {1} ; \&
	\node [cell] {$\efrac{1}{6}$} ; \&
	\node [cell] {$\efrac{1}{6}$} ; \&
	\node [cell] {$\efrac{1}{6}$} ; \&
	\\
	\node [cell] (t2n1) {$\efrac{1}{6}$} ; \&
	\node [cell] {$\efrac{1}{6}$} ; \&
	\node [cell] {$\efrac{1}{6}$} ; \&
	\node [cell] {$\efrac{1}{6}$} ; \&
	\node [cell] {$\efrac{1}{6}$} ; \&
	\\
	\node [cell] {$\efrac{1}{6}$} ; \&
	\node [cell] {$\efrac{1}{6}$} ; \&
	\node [cell] (t4n1) {1} ; \&
	\node [cell] {1} ; \&
	\node [cell] {1} ; \&
	\\
	\node [cell] (anc) {$\efrac{1}{6}$} ; \&
	\node [cell] (t2n2) {$\efrac{1}{6}$} ; \&
	\node [cell] {1} ; \&
	\node [cell] {1} ; \&
	\node [cell] (t34n2) {1} ; \&
	\\
};
\begin{pgfonlayer}{background}
\node[tile, fit=(t5n1) (t34n2), draw=black!70, fill=black!10]  {};
\node[tile, fit=(t1n1) (t1n2), draw=orange!70, fill=orange!30]  {};
\node[tile, fit=(t4n1) (t34n2), draw=red!70, fill=red!30]  {};
\end{pgfonlayer}
\end{tikzpicture}}%
\subfigure[$\ifam{U} \cup \ifam{B}$]{
\begin{tikzpicture} 
\matrix[ampersand replacement = \&] {
	\node [cell] (t1n1) {1} ; \&
	\node [cell] {1} ; \&
	\node [cell] {$\efrac{1}{2}$} ; \&
	\node [cell] {$\efrac{1}{2}$} ; \&
	\node [cell] {$\efrac{1}{2}$} ; \&
	\\
	\node [cell] (t5n1) {1} ; \&
	\node [cell] (t1n2) {1} ; \&
	\node [cell] {$\efrac{1}{3}$} ; \&
	\node [cell] {$\efrac{1}{3}$} ; \&
	\node [cell] {$\efrac{1}{3}$} ; \&
	\\
	\node [cell] (t2n1) {0} ; \&
	\node [cell] {0} ; \&
	\node [cell] {$\efrac{1}{3}$} ; \&
	\node [cell] (t3n1) {1} ; \&
	\node [cell] {1} ; \&
	\\
	\node [cell] {0} ; \&
	\node [cell] {0} ; \&
	\node [cell] {$\efrac{1}{3}$} ; \&
	\node [cell] {1} ; \&
	\node [cell] {1} ; \&
	\\
	\node [cell] (anc) {0} ; \&
	\node [cell] (t2n2) {0} ; \&
	\node [cell] {$\efrac{1}{3}$} ; \&
	\node [cell] {1} ; \&
	\node [cell] (t34n2) {1} ; \&
	\\
};
\begin{pgfonlayer}{background}
\node[tile, fit=(t5n1) (t34n2), draw=black!70, fill=black!10]  {};
\node[tile, fit=(t1n1) (t1n2), draw=orange!70, fill=orange!30]  {};
\node[tile, fit=(t2n1) (t2n2), draw=blue!70, fill=blue!30]  {};
\node[tile, fit=(t3n1) (t34n2), draw=green!70, fill=green!30]  {};
\end{pgfonlayer}
\end{tikzpicture}}%
\subfigure[$\ifam{T} \cup \ifam{U} \cup \ifam{B}$\label{fig:toy:model5}]{
\begin{tikzpicture} 
\matrix[ampersand replacement = \&] {
	\node [cell] (t1n1) {1} ; \&
	\node [cell] {1} ; \&
	\node [cell] {$\efrac{1}{2}$} ; \&
	\node [cell] {$\efrac{1}{2}$} ; \&
	\node [cell] {$\efrac{1}{2}$} ; \&
	\\
	\node [cell] (t5n1) {1} ; \&
	\node [cell] (t1n2) {1} ; \&
	\node [cell] {0} ; \&
	\node [cell] {0} ; \&
	\node [cell] {0} ; \&
	\node [cell] {$\vphantom{\efrac{1}{2}}$}; 
	\\
	\node [cell] (t2n1) {0} ; \&
	\node [cell] {0} ; \&
	\node [cell] {0} ; \&
	\node [cell] (t3n1) {1} ; \&
	\node [cell] {1} ; \&
	\node [cell] {$\vphantom{\efrac{1}{2}}$}; 
	\\
	\node [cell] {0} ; \&
	\node [cell] {0} ; \&
	\node [cell] (t4n1) {1} ; \&
	\node [cell] {1} ; \&
	\node [cell] {1} ; \&
	\node [cell] {$\vphantom{\efrac{1}{2}}$}; 
	\\
	\node [cell] (anc) {0} ; \&
	\node [cell] (t2n2) {0} ; \&
	\node [cell] {1} ; \&
	\node [cell] {1} ; \&
	\node [cell] (t34n2) {1} ; \& 
	\node [cell] {$\vphantom{\efrac{1}{2}}$};
	\\
};
\begin{pgfonlayer}{background}
\node[tile, fit=(t5n1) (t34n2), draw=black!70, fill=black!10]  (t1) {};
\node[tile, fit=(t1n1) (t1n2), draw=orange!70, fill=orange!30]  {};
\node[tile, fit=(t2n1) (t2n2), draw=blue!70, fill=blue!30]  {};
\node[tile, fit=(t3n1) (t34n2), draw=green!70, fill=green!30]  {};
\node[tile, fit=(t4n1) (t34n2), draw=red!70, fill=red!30]  {};

\end{pgfonlayer}
\end{tikzpicture}}
\caption{Toy example of a dataset, and maximum entropy models for this data given different sets of tiles, as used in Examples~\ref{ex:toy1}--\ref{ex:toy5}.
({a}) toy dataset $D$ and five tiles, $T_1,\ldots,T_5$, identified therein, 
({b}) the maximum entropy model for $D$ given tile set $\mathcal{T} = \{T_2, T_4\}$,
({c}) the maximum entropy model for $D$ given tile set $\mathcal{U} = \{T_2,T_3,T_5\}$, 
({d}) the maximum entropy model for $D$ given $\mathcal{T} \cup \mathcal{U}$, 
({e}) the maximum entropy model for $D$ given $\mathcal{T} \cup \mathcal{B}$, with $\mathcal{B} = \{T_1\}$, 
({f}) the maximum entropy model for $D$ given tile set $\mathcal{U} \cup \mathcal{B}$, and 
({g}) the maximum entropy model for $D$ given tile set $\mathcal{T}\cup\mathcal{U}\cup\mathcal{B}$.
The cell entries in (b)--(g) are given by Theorem~\ref{thr:exponential}.
}
\label{fig:toy}
\end{figure}

A classic theorem states that the Maximum Entropy distribution $\pemp$ can be written in an exponential form---which ensures generality, and as we show below, provides us with a practical way for inferring the model.

\begin{theorem}[Theorem 3.1 in~\citep{csiszar:75:i-divergence}]
\label{thr:maxent}
Given a tile set $\ifam{T}$, a distribution $\pemp$ is the
maximum entropy distribution if and only if it can be written as
\[
	\pemp(D) \propto
	\begin{cases}
		\exp\fpr{\sum_{T \in \ifam{T}} \lambda_T \abs{\area{T}}\freq{T; D}} & D \notin \mathcal{Z} \\
		0 & D \in \mathcal{Z}, \\
	\end{cases}
\]
where $\lambda_T$ is a certain weight for $\freq{T; D}$ and $\mathcal{Z}$ is a collection of datasets
such that $p(D) = 0$ for each $p \in \mathcal{P}$.
\end{theorem}

The next theorem allows to factorize the distribution $\pemp$ into a product of Bernoulli
random variables, each variable representing a single entry in the dataset. Such a representation
gives us a practical way for inferring the model.

\begin{theorem}
\label{thr:exponential}
Let $\ifam{T}$ be a tile set. Write $\ifam{T}(i, j) = \set{T \in \ifam{T} \mid (i, j) \in \area{T}}$ 
to be the subset of $\ifam{T}$ containing the tiles that cover an entry $(i, j)$.
Then, the maximum entropy distribution can be factorised as $\pemp(D) = \prod_{i, j} \pemp((i, j) = D(i, j))$,
where
\[
	\pemp((i, j) = 1) = \frac{\exp\fpr{\sum_{T \in \ifam{T}(i, j)} \lambda_T}}{\exp\fpr{\sum_{T \in \ifam{T}(i, j)} \lambda_T} + 1} \quad \text{or} \quad \pemp((i, j) = 1) = 0, 1 \quad .
\]
\end{theorem}
\begin{proof}
See Appendix~\ref{sec:apx}. \qed
\end{proof}

Theorem~\ref{thr:exponential} allows to represent $\pemp$ as Bernoulli
variables.  We should stress that this is a different model than assuming
independence between items in a random transaction.  Our next step is to
discover the correct frequencies for these variables.  Here we use a variant of
a well-known Iterative Scaling algorithm~\citep{darroch:72:generalized}. The
algorithm is given in Algorithm~\ref{alg:iterscale}. Informally said, given a
tile $T$, the algorithm updates the probabilities in the model such that the expected frequency of $1$s within the area defined by $T$ becomes equal to $\alpha_T$. We do this for each tile. However, updating the model to satisfy a single tile might change the frequency of another tile, and hence we need several passes. The update phase is done by modifying the weight $\lambda_T$ (given in
Theorem~\ref{thr:maxent}). For the proof of convergence see Theorem~3.2~in~\citep{csiszar:75:i-divergence}.

Instead of representing a distribution with the exponential form given in
Theorem~\ref{thr:maxent} we prefer working with the matrix form given in
Theorem~\ref{thr:exponential}. We need to show that the update phase in
Algorithm~\ref{alg:iterscale} actually corresponds modifying the weight
$\lambda_T$. In order to do that let $p$ be a distribution having the
exponential form, let $T$ be a tile let $q$ be a distribution $q(D) \propto
\exp\fpr{\lambda \freq{T; D}}p(D)$, that is $q$ is the distribution for which
$\lambda_T$ is shifted by $\lambda$.
Let us define
\begin{equation}
\label{eq:update}
	u(y, x) = \frac{xy}{1 - y(1 - x)}\quad.
\end{equation}
For the sake of clarity, let us write $a = p((i, j) = 1)$ and $b = q((i, j) = 1)$.
We claim that $b = u(a, e^\lambda)$ if $(i, j) \in \area{T}$ and
$a = b$ otherwise.
The latter claim follows directly from Theorem~\ref{thr:exponential} since if
$(i, j) \notin \area{T}$, then the fraction in Theorem~\ref{thr:exponential}
does not contain $\lambda_T$. Assume that $(i, j) \in \area{T}$. Then  
\[
	b = \frac{e^\lambda c}{e^\lambda c + 1} \quad\text{ and }\quad a = \frac{c}{c + 1},\quad \text{ where }\quad c = \sum_{T \in \ifam{T}(i, j)} \lambda_T\quad.
\]
Expressing $b$ using $a$ implies that $b = u(a, e^\lambda)$.
In order to find the correct $e^\lambda$ we use interpolation search using a linear interpolation.

\begin{algorithm}[tb!]
\Input{tile set $\ifam{T}$, target frequencies $\set{\alpha_T}$}
\Output{Maximum entropy distribution $p$}
$p \define $ a \by{n}{m} matrix with values $1/2$\;
\lForEach{$T \in \ifam{T}$, $\alpha_T = 0, 1$} {
	$p(i, j) \define \alpha_T$ for all $(i, j) \in \area{T}$\;
}
\While {not converged} {
\ForEach{$T \in \ifam{T}$, $0 < \alpha_T < 1$} {
		$f \define \freq{T; p}$\;
		find $x$ such that $f = \sum_{(i, j) \in \area{T}} u(p(i, j), x)$\tcpas*{see Equation~\ref{eq:update}}
		$p(i, j) \define u(p(i, j), x)$ for all $(i, j) \in \area{T}$\;
	}

}
\caption{Iterative Scaling for solving the MaxEnt distribution}
\label{alg:iterscale}
\end{algorithm}

It turns out, that if the given tiles are exact, the maximum entropy
distribution has a very simple form. The Bernoulli variable corresponding to the $(i, j)^{\mathrm{th}}$ entry is always $1$ (or $0$), if that entry is covered by an exact tile, i.e. with frequency $1$ (or $0$). Otherwise, the variable is equal to a fair coin toss.
This form will allow us to express distances between sets of exact tiles in the next section.

\begin{theorem}
\label{thr:exactfactor}
Let $\ifam{T}$ be a collection of exact tiles and let $\alpha_T$ be the desired frequency of a tile $T \in \ifam{T}$. Then
\[
	\pemp_\ifam{T}((i, j) = 1) =
	\begin{cases}
		\alpha_T & \text{if there exists } T \in \ifam{T} \text{ such that } (i, j) \in \area{T} \\
		1/2 & \text{otherwise} \quad .
	\end{cases}
\]
\end{theorem}
\begin{proof}
See Appendix~\ref{sec:apx}. \qed
\end{proof}

In essence, this theorem states that in the maximum entropy model for a dataset given set $\ifam{T}$ of exact tiles, the probability $p^*$ of observing a $1$ in a specified cell $(i,j)$ is either $0$, or $1$ if this cell maps to one of the tiles in $\ifam{T}$---and otherwise, the probability is $1/2$ as we cannot infer anything about this cell from $\ifam{T}$.

\begin{example}
\label{ex:toy2}
Let us continue Example~\ref{ex:toy1}. Consider the following three tile sets $\ifam{T} = \set{T_2, T_4}$, $\ifam{U} = \set{T_2, T_3, T_5}$, and $\ifam{B} = \set{T_1}$.
The corresponding maximum entropy models are given in Figure~\ref{fig:toy}, such that each entry represents the probability $\pemp((i, j) = 1)$.
As the sets $\ifam{T}$ and $\ifam{U}$ contain only exact tiles,
Theorem~\ref{thr:exactfactor} implies that the entries for the models $\pemp_{\ifam{T}}$,
$\pemp_{\ifam{U}}$, and $\pemp_{\ifam{T} \cup \ifam{U}}$ are either $0$, $1$,
or $1/2$.

Consider $\pemp_{\ifam{T} \cup \ifam{B}}$. Corollary~\ref{cor:exact} states
that entries in $\area{T_2}$ and $\area{T_4}$ should be $1$, since both tiles
are exact. From $T_1$, we know that there are $10$ ones, yet
$T_2$ and $T_4$ account only for $8$. Hence, there should be $2$ ones in the $12$ entries outside of $\ifam{T}$, on average. We
aim to be as fair as possible, hence we spread uniformly, giving us probabilities $2/12 = 1/6$ for the unaccounted entries in $T_1$.
For $\pemp_{\ifam{U} \cup \ifam{B}}$, there are $2$ unaccounted $1$s in $6$ entries, giving us $2/6 = 1/3$. Finally, all $1$s are accounted 
for in $\pemp_{\ifam{T} \cup \ifam{U} \cup \ifam{B}}$. Outside of $\ifam{T} \cup \ifam{U} \cup \ifam{B}$ we have no information, hence these probabilities default to $1/2$.
\end{example}

\section{Comparing Sets of Tiles}\label{sec:compare}

Now that we have a technique for incorporating the information
contained within a set of tiles into a probabilistic model, we can subsequently use these models to compare between the sets of tiles.

We assume that we are given three tile sets $\ifam{T}$, $\ifam{U}$, and
$\ifam{B}$.  Let tile set $\ifam{B}$ contain the background information.
Typically, this information would be simple, like column margins, row margins,
or just the proportions of ones in the whole dataset. $\ifam{B}$ can be also
be empty, if we do not have or wish to use any background knowledge. Our goal is now to compute the distance between $\ifam{T}$ and $\ifam{U}$ given $\ifam{B}$.  We assume that the frequencies for the tile sets we are given are mutually consistent; which is automatically guaranteed if the frequencies for all three tile sets are
computed from a single dataset. Now, let $\ifam{M} = \ifam{T} \cup \ifam{U} \cup
\ifam{B}$ be the collection containing all tiles. We define the distance between $\ifam{T}$ and $\ifam{U}$, w.r.t. $\ifam{B}$, as
\[
	\dist{\ifam{T}, \ifam{U} ; \ifam{B}} = \frac{\kl{\ifam{M}}{\ifam{U} \cup \ifam{B}} + \kl{\ifam{M}}{\ifam{T} \cup \ifam{B}}}{\kl{\ifam{M}}{\ifam{B}}} \quad .
\]
If $\kl{\ifam{M}}{\ifam{B}} = 0$, we define $\dist{\ifam{T}, \ifam{U} ; \ifam{B}} = 1$.
A direct application of Theorem~\ref{thr:exponential} leads to the following theorem which in turns imply that we can compute the distance in $O(nm)$ time. 
\begin{theorem}
\label{thr:compute}
Let $\ifam{A}$ and $\ifam{B}$ be two sets of tiles such that $\ifam{B} \subseteq \ifam{A}$. 
Then
\[
	\kl{\ifam{A}}{\ifam{B}} = \sum_{i, j} \sum_{v = 0, 1} \pemp_\ifam{A}((i, j) = v) \log \frac{\pemp_\ifam{A}((i, j) = v)}{\pemp_\ifam{B}((i, j) = v)}\quad.
\]
\end{theorem}

If the given tiles are exact, the distance has a simple interpretable form; namely, the distance can be expressed with Jaccard similarity.

\begin{theorem}
\label{thr:distanceexact}
Assume three tile collections $\ifam{T}$, $\ifam{U}$, and $\ifam{B}$ with exact
frequencies $\set{\alpha_T}$, $\set{\beta_U}$, and $\set{\gamma_B}$.
Define $X = \area{\ifam{T}} \setminus \area{\ifam{B}}$ and $Y = \area{\ifam{U}} \setminus \area{\ifam{B}}$.
Then $\dist{\ifam{T}, \ifam{U} ; \ifam{B}} = 1 - \abs{X \cap Y}/\abs{X \cup Y}$ for $\abs{X \cup Y} > 0$ and $\dist{\ifam{T}, \ifam{U} ; \ifam{B}} = 1$
for $\abs{X \cup Y} = 0$.
\end{theorem}
\begin{proof}
See Appendix~\ref{sec:apx}. \qed
\end{proof}

In other words, the theorem states that if we only consider sets of exact tiles, we can easily compute (and interpret!) the dissimilarity score $d$ between them, as our measure coincides with Jaccard dissimilarity.

\begin{example}
\label{ex:toy3}
Let us continue Example~\ref{ex:toy2}. To compute $\dist{\ifam{T}, \ifam{U} ;
\emptyset}$ we first note that $\ifam{T}$ and $\ifam{U}$ only have exact tiles,
and hence we can use Theorem~\ref{thr:distanceexact}. So, we have 
$\abs{\area{\ifam{T}}  \setminus \area{\ifam{U}}} = 2$, $\abs{\area{\ifam{U}}  \setminus \area{\ifam{T}}} = 8$, and $\abs{\area{\ifam{T}}  \cup \area{\ifam{U}}} = 18$.
And hence, the distance $\dist{\ifam{T}, \ifam{U} ; \emptyset} = (2 + 8) / 18 = 5/9$.

Next, let use $\ifam{M}$ as a shorthand for $\ifam{T} \cup \ifam{U} \cup \ifam{B}$, i.e. $\ifam{M} = \ifam{T} \cup \ifam{U} \cup \ifam{B}$.
To compute $\dist{\ifam{T}, \ifam{U} ; \ifam{B}}$, note that 
\[
	\kl{\ifam{M}}{\ifam{T} \cup \ifam{B}} = 2\log(6) + 10\log(6/5) \approx 5.4067 \quad .
\]
where the first term represents the positive entries in $\ifam{M}$  and the second term the negative entries in $\ifam{M}$.
Similarly, $\kl{\ifam{M}}{\ifam{B}} \approx 15.2$, and $\kl{\ifam{M}}{\ifam{U} \cup \ifam{B}} \approx 3.8$. Consequently, the distance between $\ifam{T}$ and $\ifam{U}$ in light of $\ifam{B}$ is
\[
\dist{\ifam{T}, \ifam{U} ; \ifam{B}} = \frac{3.8 + 5.4}{15.2} \approx 0.6
\]
which is slightly larger than $\dist{\ifam{T}, \ifam{U} ; \emptyset} \approx 0.56$. This is due to the fact
that adding $\ifam{B}$, which effectively states the data is sparse, to $\ifam{T}$ makes the probability of encountering a $1$ at $(3, 4)$ and $(3, 5)$ in the model less likely. Hence, given that background knowledge, and regarding   $\ifam{U}$, we are more surprised to find that the values of these entries are indeed ones.
\end{example}

Next, we discuss some further properties of our measure. We start by showing the measurement is upper bounded by $2$ for noisy tiles, and $1$ if we only consider exact tiles.

\begin{theorem}
\label{thr:bound}
Assume three tile collections $\ifam{T}$, $\ifam{U}$, and $\ifam{B}$. Then
$\dist{\ifam{T}, \ifam{U} ; \ifam{B}} \leq 2$.  If $\ifam{T}$, $\ifam{U}$, and
$\ifam{B}$ consist only of exact tiles, then  $\dist{\ifam{T}, \ifam{U} ; \ifam{B}} \leq 1$.
\end{theorem}
\begin{proof}
See Appendix~\ref{sec:apx}. \qed
\end{proof}

Before, in Theorem~\ref{thr:distanceexact} we showed that when all tiles are exact and $\area{\ifam{T}} \cap
\area{\ifam{U}} \subseteq \area{\ifam{B}}$, that is, $\ifam{T}$ does not provide any new information
about $\ifam{U}$, then the distance is equal to $1$. We can extend this property to noisy tiles.

\begin{theorem}
\label{thr:indcol}
Assume three tile collections $\ifam{T}$, $\ifam{U}$, and $\ifam{B}$. Assume
that $\ifam{B}$ consists of exact tiles and $\area{\ifam{T}} \cap
\area{\ifam{U}} \subseteq \area{\ifam{B}}$. Then $\dist{\ifam{T}, \ifam{U} ; \ifam{B}}  = 1$.
\end{theorem}
\begin{proof}
See Appendix~\ref{sec:apx}. \qed
\end{proof}

Exploiting the fact that our measure coincides with Jaccard dissimilarity when considering exact tiles, our measure upholds the triangle inequality when all tile sets consists of exact tiles.

\begin{theorem}
\label{thr:triangle}
Assume four tile collections $\ifam{S}$, $\ifam{T}$, $\ifam{U}$, and
$\ifam{B}$ consisting only of exact tiles. Then the distance satisfies
the triangle inequality,
\[
	\dist{\ifam{T}, \ifam{U} ; \ifam{B}} \leq \dist{\ifam{T}, \ifam{S} ; \ifam{B}} + \dist{\ifam{S}, \ifam{U} ; \ifam{B}}\quad.
\]
\end{theorem}
\begin{proof}
See Appendix~\ref{sec:apx}. \qed
\end{proof}

Consequently, the distance is (almost) a metric when working only with exact
tiles. The technical subtlety is that the distance can be 0 even if the tile
sets are different. This will happen, if $(\area{\ifam{T}} \cup
\area{\ifam{U}}) \setminus (\area{\ifam{T}} \cap \area{\ifam{U}}) \subseteq
\area{\ifam{B}}$. This highlights the fact that we are measuring the
information hidden in the tile sets, not the actual tiles.

Next, we show that if we add an element from one tile set to the other, the distance we measure between them decreases.

\begin{theorem}
\label{thr:add}
Assume three tile collections $\ifam{T}$, $\ifam{U}$, and $\ifam{B}$. Let $U \in \ifam{U}$.
Then
\[
	\dist{\ifam{T} \cup \set{U}, \ifam{U} ; \ifam{B}} \leq \dist{\ifam{T}, \ifam{U} ; \ifam{B}}\quad.
\]
\end{theorem}
\begin{proof}
See Appendix~\ref{sec:apx}. \qed
\end{proof}

In other words, if we add tiles from the right-hand tile set $\ifam{U}$ to the left-hand tile set $\ifam{T}$, regardless of the background information $\ifam{B}$, we decrease the dissimilarity $\dist{\ifam{T},\ifam{U} ; \ifam{B}}$, as $\ifam{T}$ contains more information about $\ifam{U}$.

\begin{example}
\label{ex:toy4}
Consider the tile sets $\ifam{U}$ and $\ifam{M} = \ifam{T} \cup \ifam{U} \cup \ifam{B}$ given in Figure~\ref{fig:toy}.
The distance between $\ifam{U}$ and $\ifam{M}$ is $6/22$.
The distance between $\ifam{U} \cup \set{T_4} = \ifam{U} \cup \ifam{T}$ and $\ifam{M}$ is $3/22$.
Since $T_4 \in \ifam{M}$, Theorem~\ref{thr:subset} states that the distance $\dist{\ifam{U}, \ifam{M}; \emptyset}$
should be as large as $\dist{\ifam{U} \cup \set{T_4}, \ifam{M}; \emptyset}$, which is the case.
\end{example}

If $\ifam{U} \subseteq \ifam{T}$, then we can show that the smaller the distance,
the larger the Kullback-Leibler is between $\ifam{U}$ and the background
knowledge. We will use this property for iterative data mining in Section~\ref{sec:iter}. 

\begin{theorem}
\label{thr:subset}
Assume three tile collections $\ifam{T}$, $\ifam{U}$, and $\ifam{B}$. Let
$\ifam{M} = \ifam{T} \cup \ifam{U} \cup \ifam{B}$. Assume that $\ifam{U} \cup
\ifam{B} \subseteq \ifam{T} \cup \ifam{B}$.
Then
\[
	\dist{\ifam{T}, \ifam{U} ; \ifam{B}} = \frac{\kl{\ifam{M}}{\ifam{U} \cup \ifam{B}}}{\kl{\ifam{M}}{\ifam{B}}} = 1 - \frac{\kl{\ifam{U} \cup \ifam{B}}{\ifam{B}}}{\kl{\ifam{M}}{\ifam{B}}} \quad.
\]
\end{theorem}
\begin{proof}
See Appendix~\ref{sec:apx}. \qed
\end{proof}

Let us next consider an example of this property.

\begin{example}
\label{ex:toy5}
Consider the tile sets $\ifam{T}$ and $\ifam{U}$ given in Figure~\ref{fig:toy}.
Let $\ifam{M} = \ifam{T} \cup \ifam{U}$.
We have $\kl{\ifam{M}}{\ifam{U}} = 2$, $\kl{\ifam{U}}{\emptyset} = 16$, $\kl{\ifam{M}}{\emptyset} = 18$.
The distance is then
\[
	\dist{\ifam{U}, \ifam{U} \cup \ifam{T} ; \emptyset}
		= \frac{\kl{\ifam{M}}{\ifam{U}}}{\kl{\ifam{M}}{\emptyset}} = \frac{2}{18} = 1 - \frac{16}{18} = 1 - \frac{\kl{\ifam{U}}{\emptyset}}{\kl{\ifam{M}}{\emptyset}} \quad.
\]
\end{example}

\section{Applying the distance}\label{sec:appl}

Above, we were only concerned in finding out how much information two tile sets share. With our measure, we can quantify the difference in information provided by different tile sets, and hence, measure differences between results---a very useful application on itself. In this section we consider two more elaborate problems to which our measure can be applied. 

\subsection{Redescribing Results by Redescribing Sets of Tiles}\label{sec:redescr}

The first application we consider is redescription of results. Opposed to traditional redescription mining~\citep{ramakrishnan:04:turning}, where given a pattern $X$ we want to find a syntactically different pattern $Y$ that identifies the same rows in the database as $X$, we here study how well we can describe the information provided by one result using (parts of) another result. Between these two approaches, the underlying question is the same: are there different representations of the same information. Whereas in standard redescription mining the row-set of a pattern is regarded as `information', we consider the detail given on the data (distribution) as information.

More formally, given two tile sets, say $\ifam{T}$ and $\ifam{U}$, we want to find which tiles from $\ifam{U}$ best describe the information provided by $\ifam{T}$. To this end, we can apply our distance measure as follows.

\begin{problem}[\textsc{Redescribe}]
Given three sets of tiles $\ifam{T}$, $\ifam{U}$, and $\ifam{B}$ with consistent frequencies,
find a subset $\ifam{V} \subseteq \ifam{U}$ such that $\dist{\ifam{V}, \ifam{T}; \ifam{B}}$
is minimised.
\end{problem}

As with many optimisation problems in data mining, it turns out that finding the best tile subset is computationally intractable. 

\begin{theorem}
\label{thr:np}
The decision version of \textsc{Redescribe} is an \np-hard problem.
\end{theorem}
\begin{proof}
See Appendix~\ref{sec:apx}. \qed
\end{proof}

Hence, we resort to a simple greedy heuristic to find good approximations for the redescription problem: we iteratively choose to add that tile to our selection that maximally reduces the difference to the target tile set. We give the pseudo-code for our redescription mining approach as Algorithm~\ref{alg:fruits}.
As the number of tiles we can select is finite, so is the procedure. Moreover, we stop selecting new tiles when we cannot decrease the distance by adding more tiles.

\begin{algorithm}[ht!]
\caption{\textsc{Fruits}, finds redescription used in a tile set. }\label{alg:fruits}
\Input{target tile set $\ifam{T}$, candidate tile set $\ifam{S}$, background knowledge $\ifam{B}$}
\Output{$\ifam{R}$ a subset of $\ifam{S}$ redescribing $\ifam{T}$}

$\ifam{R} \define \emptyset$\;
\While {changes} {
	$d \define \dist{\ifam{R}, \ifam{T} ; \ifam{B}}$\;
	$C \define \emptyset$\;
	\ForEach{$S \in \ifam{S}$} {
		\If {$\dist{\ifam{R} \cup \set{C}, \ifam{T} ; \ifam{B}} \leq d$} {
			$C \define S$\;
			$d \define \dist{\ifam{R} \cup \set{C}, \ifam{T} ; \ifam{B}}$\;
		}
	}
	\lIf {$C \neq \emptyset$} {add $C$ to $\ifam{R}$; remove $C$ from $\ifam{S}$;}
}
\Return $\ifam{R}$\;
\end{algorithm}

Note that Algorithm~\ref{alg:fruits} selects tiles in a particular order: we iteratively find the tile that minimises the distance to the information in the target tile set. As such, if needed, one can only inspect the top-$k$ selected tiles of a redescription to see how result $\ifam{S}$ captures the most important information in $\ifam{T}$.

Furthermore, by iteratively minimising the distance, if the candidate tile set $\ifam{S}$ contains elements of the target tile set $\ifam{T}$, these tiles are likely to be selected---as those parts of the target information can be `redescribed' exactly. Also note that by minimising the distance between $\ifam{R}$ and $\ifam{T}$ we will not simply select all tiles in $\ifam{S}$ that overlap with tiles in $\ifam{T}$---even though by overlapping these tiles will provide some similar information, their non-overlapping parts will provide different information, and hence increase the distance between the two tile sets. As such, we automatically prevent the selection of tiles in $\ifam{S}$ that are much more specific (e.g. contain more items) than what $\ifam{T}$ describes.

\subsection{Iterative Data Mining}\label{sec:iter}

Next, we apply our distance measure for application in an iterative data mining setting.  The key idea of iterative data mining is that what one finds interesting depends on what one knows. As such, practitioners study results one at the time~\citep{hanhijarvi:09:tell,mampaey:12:tkdd}. Once we have studied a result, we are no longer interested in other results that provide (approximately) the same information, and we rather want to be served the most interesting pattern with regard to what we now know.

As such, the key difference between iterative data mining and traditional pattern ranking is that we take into account what a user has already seen when determining what might be interesting. Consequently, if a result has a peculiar statistic, yet it can be (approximately) explained by the previously inspected results, we consider this result to be redundant and choose not present it to the user.
To be more specific, let us assume that we ran $k$ different algorithms on a
single dataset $D$, and so obtained $\ifam{T}_1, \ldots, \ifam{T}_k$ different tile sets. Now, if we merge all tiles into one tile set $\ifam{T} = \ifam{T}_1 \cup \cdots \cup \ifam{T}_k$, this set represents all information we have assembled about $D$. Note however that $\ifam{T}$ is unordered, and will likely  contain redundant tiles. We order the tiles by applying \textsc{Fitamin}, which orders the tiles in a given tile set $\ifam{T}$ such that every tile provides maximal novel information with regard to the tiles ranked above it.

\begin{algorithm}[ht!]
\caption{\textsc{Fitamin}, Find iteratively those tiles that are most informative}
\label{alg:fitamin}
\Input{tile set $\ifam{T}$, background knowledge $\ifam{B}$}
\Output{ordered list of tiles occurring in $\ifam{T}$ }

$\ifam{L} \define \emptyset$\;
$\ifam{U} \define \ifam{T}$\;

\While {$\ifam{U} \neq \emptyset$} {
	$T \define \arg \min_U \set{\dist{\ifam{L} \cup \set{U}, \ifam{T}; \ifam{B}} \mid U \in  \ifam{U}}$\;
	add $T$ to the end of $\ifam{L}$\;
	delete $T$ from $\ifam{U}$\;
}

\Return $\ifam{L}$\;
\end{algorithm}

The \textsc{Fitamin} algorithm, for which we present the pseudo-code as Algorithm~\ref{alg:fitamin}, starts with an empty list $\ifam{L}$ and selects
a tile $T \in \ifam{T}$ which minimises the distance $\dist{\ifam{L} \cup \set{U}, \ifam{T};
\ifam{B}}$.  We add $T$ into $\emph{L}$ and repeat the process until all tiles
are ordered.

In order to understand the connection between this approach and iterative data
mining, let us consider Theorem~\ref{thr:subset} which implies that the tile
that minimises the distance, also maximises $\kl{\ifam{L} \cup \ifam{B} \cup \set{T}}{\ifam{B}}$.
Corollary~\ref{cor:diff} (given in Appendix) allows us to rewrite the divergence
as a difference between two entropies
\[
	\kl{\ifam{L} \cup \ifam{B} \cup \set{T}}{\ifam{B}} = \ent{\ifam{B}} - \ent{\ifam{L} \cup \ifam{B} \cup \set{T}}\quad.
\]
Note that $\ent{\ifam{B}}$ does not depend on $T$ and so if we replace it with $\ent{\ifam{L} \cup \ifam{B}}$,
then we conclude that $T$ must maximise
\[
	\ent{\ifam{B} \cup \ifam{L}} - \ent{\ifam{L} \cup \ifam{B} \cup \set{T}} = \kl{\ifam{L} \cup \ifam{B} \cup \set{T}}{\ifam{B} \cup \ifam{L}},
\]
where the equality follows from Corollary~\ref{cor:diff}.

The divergence $\kl{\ifam{L} \cup \ifam{B} \cup \set{T}}{\ifam{B} \cup
\ifam{L}} = 0$ if and only if the frequency of $T$ obtained from the dataset is
equal to the frequency obtained from the maximum entropy model constructed from
the known tiles $\ifam{B} \cup \ifam{L}$. The divergence increases as these two
frequencies become more different. Consequently, a tile $T$ is the one that is
most surprising given the known tiles.

\section{Related Work}\label{sec:related}

To our knowledge, defining a distance between two \emph{general} tile sets is a
novel idea.  However, there exist several techniques for comparing between datasets by comparing how the supports of patterns change between the datasets.  Such proposals include a Mahalanobis distance
between itemset collections~\citep{tatti:07:distances} and a compression-based
distance between itemsets~\citep{vreeken:07:difference}. In addition, \cite{hollmen:03:mixture} suggested 
using $L_1$ distance between frequent itemset collections, where the
missing frequencies were estimated by a support threshold.

From technical point of view, comparing pattern sets given background
knowledge is akin to defining an interestingness measure based on deviation
from the background knowledge.  In fact, our approach for building a global
Maximum Entropy model from tiles was inspired by the work of \cite{debie:11:dami}, where he builds a similar maximum entropy model from row and column margins 
(i.e. a Rasch model~\citep{rasch:60:probabilistic}) 
and uses it as a static null hypothesis to rank tiles. Further related proposals include iterative mining of patterns by empirical $p$-values and randomisation~\citep{hanhijarvi:09:tell}, and maximum entropy models
based on itemsets~\citep{wang:06:summaxent, mampaey:12:tkdd}.

Several techniques have been proposed for mining sets of tiles. \cite{geerts:04:tiling}
suggested discovering tilings that cover as many ones as
possible. \cite{xiang:10:hyper} gave a method to mine (possibly noisy) tiles that cover ones while minimising a cost: the number of
transactions and items needed to describe tiles.
These methods focus on covering the ones in the data, alternatively, we can
assess the quality of a tiling by statistical means.
\cite{gionis:04:geometric} suggested discovering hierarchical tiles by building a statistical model and optimising an MDL score. 
\cite{debie:11:dami} gave a maximum entropy model based on column/row margins to rank tiles. \cite{vreeken:11:krimp} propose that the best set of tiles (or itemsets) is the tile set that compresses the dataset best.

An alternative approach for discovering tiles is to consider a Boolean matrix
factorisation~\citep{miettinen:08:discrete}. That is, factorise the dataset into two low rank Boolean matrices, where the row vectors of the one matrix correspond to itemsets, while the column vectors of the other matrix correspond to \emph{tid}-lists. The Boolean product of these matrices naturally defines a set of noisy tiles.

Compared to tiles, computing maximum entropy models based on itemsets is much
more difficult. The reason for this is that there is no equivalent version of
Lemma~\ref{lem:alternative} for itemsets. In fact, computing an expected value
of an itemset from a maximum entropy model is
\textbf{PP}-hard~\citep{tatti:06:computational}.  To avoid these problems, we
can convert itemsets to exact tiles by considering their supporting transactions.

Redescribing tile sets is closely related to redescription
mining, in which the idea is to find pairs of syntactically different patterns covering roughly the same transactions. 
\cite{ramakrishnan:04:turning} originally approached the
problem by building decision trees. Other approaches include
Boolean Formulae with no overlap~\citep{gallo:08:finding}, and exact minimal redescriptions~\citep{zaki:05:reasoning}. From a computational point of view, the difference between our problem and existing work is that whereas redescription mining aims to construct alternative patterns given a \emph{single} target pattern, we on the other hand consider \emph{sets} of target and candidate patterns, and aim to find the \emph{subset} of patterns from candidates that together describe the target best.

Iterative data mining, in which one iteratively finds the most interesting result given what one already knows, and subsequently dynamically updates the background model, is a relatively new approach to data mining~\citep{debie:11:inftheoryframework}. \cite{hanhijarvi:09:tell} propose to use swap-randomisation for acquiring empirical p-values for patterns, keeping the margins of already selected patterns (approximately) fixed. By using an analytical model of the data, instead of having to sample many randomised datasets per pattern, we gain much computational efficiency. Whereas we model complete databases, \cite{mampaey:12:tkdd} formalise a maximum entropy model for rows of the data by iteratively choosing that itemset for which the frequency prediction is most off. Both models have advantages: whereas in our model it is not obvious how itemset frequencies can be easily incorporated, i.e. without specifying in which rows of the dataset these occur, in the row-model approach co-occurrences of patterns are difficult to detect, as we do not know \textit{where} patterns occur.

\section{Experiments}
\label{sec:exps}

In this section we empirically evaluate our measure, applying it for measuring distances between results, redescriptions of results, and for iterative data mining. 

\subsection{Set up}

We evaluate our measure on four publicly available real world datasets.
The {\em Abstracts} dataset contains the abstracts of the papers accepted at ICDM up to 2007, where words have been stemmed and stop words removed~\citep{debie:11:dami}.
The {\em DNA} amplification data is data on DNA copy number amplifications. Such copies activate oncogenes and are hallmarks of nearly all advanced tumours~\citep{myllykangas:06:dna}.
Amplified genes represent attractive targets for therapy, diagnostics and 
prognostics.
The {\em Mammals} presence data consists of presence records of European mammals\footnote{
The full version of the mammal dataset is available for research 
purposes upon request from the Societas Europaea Mammalogica.~
\url{http://www.european-mammals.org}} within geographical areas of $50 \times 50$ kilometers~\citep{mitchell-jones:99:atlas}.
Finally, {\em Paleo} contains information on fossil records\footnote{NOW public release 030717 available
from~\citep{fortelius:06:spectral}.}
found at specific palaeontological sites in Europe~\citep{fortelius:06:spectral}.

We provide our prototype implementation for research purposes\footnote{\url{http://www.adrem.ua.ac.be/implementations/}}.
Computing a single distance typically takes a few seconds---up to maximally two minutes for the \emph{Abstracts} dataset, when considering the most complex sets of tiles, and most detailed background knowledge.

\subsection{Methods and Mining Results}
\begin{table}[tb!]
\centering
\setlength{\tabcolsep}{5pt}
\begin{tabular*}{\linewidth}{@{\extracolsep{\fill}}l rr rrrrrrrrrr}
\toprule
&&& \multicolumn{10}{l}{\textbf{Number of Tiles per Method}}\\
\cmidrule{4-13}
\textbf{Dataset} & $n$ & $m$ & \multicolumn{1}{c}{\emph{clus}} & \multicolumn{1}{c}{\emph{bicl}} & \multicolumn{1}{c}{\emph{atcl}} & \multicolumn{1}{c}{\emph{sscl}} & \multicolumn{1}{c}{\emph{asso}} & \multicolumn{1}{c}{\emph{tiling}} & \multicolumn{1}{c}{\emph{hyper}} & \multicolumn{1}{c}{\emph{itt}} & \multicolumn{1}{c}{\emph{krimp}} & \multicolumn{1}{c}{\emph{mtv}} \\
\midrule
Abstracts & $859$ & $3933$ & $5m$ & $25$ & $753$ & $100$ & $100$ & $38$ & $100$ & $100$ & $100$ & $25$ \\
DNA & $4590$ & $392$ & $5 m$ & $25$ & $56$ & $100$ & $100$ & $32$ & $100$ & $100$ & $100$ & $100$ \\
Mammals & $2183$ & $124$ & $5 m$ & $25$ & $28$ & $100$ & $91$ & $3$ & $100$ & $2$ & $100$ & $14$ \\
Paleo & $501$ & $139$ & $5 m$ & $25$ & $514$ & $100$ & $100$ & $100$ & $100$ & $71$ & $85$ & $14$ \\
\bottomrule
\end{tabular*}
\caption{Number of tiles extracted by each of the considered methods.}\label{tbl:methods} 
\vspace{-2em}
\end{table}

We apply our measure to compare between the results of ten different exploratory data mining methods for binary data. Table~\ref{tbl:methods} gives an overview, here we list the methods, stating the parameters we use, and giving how we refer to each of the methods between brackets. 

\begin{description}
\addtolength{\itemsep}{0.3\baselineskip}
\item[\textit{clus}]
We employ simple $k$-means {clus}tering with $k=5$ clusters, using $L_1$ distance. We turn the clusters into tiles by computing column margins inside each cluster. 

\item[\textit{bicl}]
A bi-clustering simultaneously clusters the transactions and the items; a cluster is then a tile defined by the corresponding pair of item and transaction clusters~\citep{pensa:05:bicluster}. 
We apply {bicl}ustering by separately clustering the columns and rows using again $k$-means clustering ($k=5$) and combine the two clusterings into a grid.
\citet{puolamaki:08:approximation} showed that a good approximation bound can be achieved with this approach.
Each cluster is represented by a single tile. 

\item[\textit{atcl}]
We use the parameter-free attribute clustering approach by~\cite{mampaey:12:summarising,mampaey:10:summarising} that groups together binary or categorical attributes for which the values show strong interaction. We convert each attribute cluster into a tile---thereby somewhat oversimplifying these results, as we disregard the identified structure in attribute-value combinations within the cluster. 

\item[\textit{sscl}]
Subspace clustering aims at finding groups of rows that exhibit strong similarity on a subset, i.e. subspace, of the attributes. For subspace clustering, we used the implementation of \citet{muller:09:sscl} of the \textsc{ProClus} algorithm~\citep{aggarwal:99:proclus}, mined $100$ clusters, each over maximally $32$ dimensions, and converted each into a noisy tile.

\item[\textit{asso}]
The \textsc{Asso} algorithm~\citep{miettinen:08:discrete} approximates the optimal $k$-rank Boolean Matrix Factorisation of the data. We ran with a maximum of $100$ factors, of which the non-empty ones were converted into tiles. 

\item[\textit{tiling}]
Using the \textsc{Tiling} algorithm~\citep{geerts:04:tiling} we mine overlapping  tilings, of up to $100$ exact tiles, allowing the algorithm a maximum running time of $8$ hours. 

\item[\textit{hyper}]
The \textsc{Hyper} algorithm is strongly related to Tiling, but can find noisy tiles by combining exact tiles such that error is minimised. Per dataset, we mined up to $100$ hyper rectangles~\citep{xiang:10:hyper}, directly obtaining  noisy tiles. 

\item[\textit{itt}]
We mined Information-Theoretic exact Tiles~\citep{debie:11:dami}, where the method automatically selects the number of tiles, and ranks them according to informativeness. Out of the returned tiles, we use the top-$100$.

\item[\textit{mtv}]
We applied the \textsc{mtv} algorithm for obtaining the most informative itemsets; those itemsets for which the predicted frequency is most off given the current model~\citep{mampaey:12:tkdd}. We allow a maximum run time of 2 hours, and a maximum number of $100$ returned itemsets. As this algorithm focuses on frequency, selecting both itemsets for which the frequency is surprisingly high as well as surprisingly low, we convert these itemsets into tiles by selecting all rows.

\item[\textit{krimp}]
We used \textsc{Krimp}~\citep{vreeken:11:krimp} to mine itemsets that compress. From the resulting code tables, we identified the top-$100$ itemsets that are most used during compression, i.e. that aid compression best, and convert these into tiles by using their \textsc{Krimp}-usage as \emph{tid}-sets.
\end{description}

The last four of these methods select itemsets from a candidate collection. As candidate itemsets to select from, we used closed frequent itemsets mined at as low as feasible support thresholds, of resp. $5$, $1$, $800$, and $1$. For \textsc{Krimp}, however, by its more optimised implementation we can use lower support thresholds for the \textit{Abstract} and \textit{Mammals} datasets, of resp. $4$ and $300$.

\subsection{Measuring Distances}

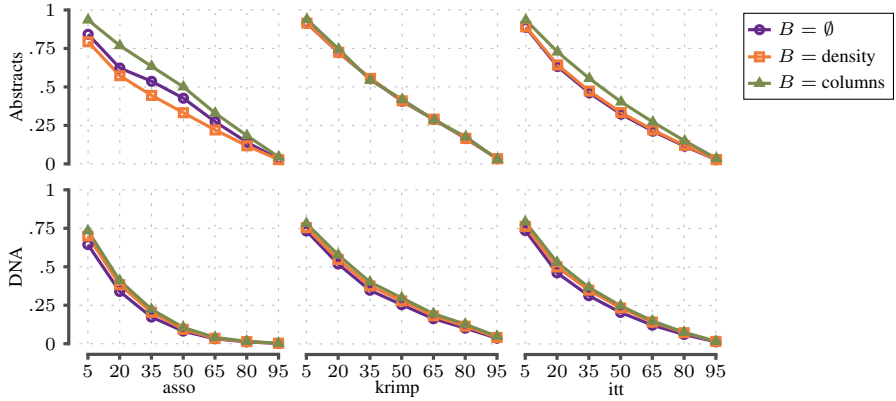
\begin{figure}[tb!]
\begin{center}
\begin{tabular}{l@{\hspace{0.2cm}}l@{\hspace{0.2cm}}l@{\hspace{0.3cm}}l}
\begin{tikzpicture}[baseline]
\begin{axis}[width = 4.2cm, xtick = {5, 20, 35, 50, 65, 80, 95}, xmin = 1, ymin = 0, ymax = 1, ytick = {0, 0.25, 0.5, 0.75, 1},
yticklabels = {$0$, $.25$, $.5$, $.75$, $1$}, 
ylabel = Abstracts, xmajorticks = false, cycle list name = yaf, legend to name = convergelegend]

\addplot coordinates {
(5, 0.842019) (20, 0.622681) (35, 0.536542) (50, 0.426695) (65, 0.273755) (80, 0.140636) (95, 0.030858) 
};
\addplot coordinates {
(5, 0.793619) (20, 0.573001) (35, 0.444870) (50, 0.332385) (65, 0.220438) (80, 0.117638) (95, 0.027454) 
};
\addplot coordinates {
(5, 0.935040) (20, 0.768148) (35, 0.633526) (50, 0.500078) (65, 0.329694) (80, 0.183515) (95, 0.044907) 
};

\legend{$B = \emptyset$, $B =$ density, $B =$ columns}

\pgfplotsextra{\yafdrawyaxis{0}{1}}
\end{axis}
\end{tikzpicture}&
\begin{tikzpicture}[baseline]
\begin{axis}[width = 4.2cm, xtick = {5, 20, 35, 50, 65, 80, 95}, xmin = 1, ymin = 0, ymax = 1, ytick = {0, 0.25, 0.5, 0.75, 1}, ticks = none, cycle list name = yaf]
\addplot coordinates {
(5, 0.913087) (20, 0.722335) (35, 0.552299) (50, 0.406771) (65, 0.287772) (80, 0.163972) (95, 0.032845) 
};
\addplot coordinates {
(5, 0.913888) (20, 0.724575) (35, 0.555084) (50, 0.409514) (65, 0.290127) (80, 0.165578) (95, 0.033270) 
};
\addplot coordinates {
(5, 0.938320) (20, 0.744917) (35, 0.543252) (50, 0.417957) (65, 0.284907) (80, 0.175545) (95, 0.029531) 
};

\end{axis}
\end{tikzpicture}&
\begin{tikzpicture}[baseline]
\begin{axis}[width = 4.2cm, xtick = {5, 20, 35, 50, 65, 80, 95}, xmin = 1, ymin = 0, ymax = 1, ytick = {0, 0.25, 0.5, 0.75, 1}, ticks = none, cycle list name = yaf, legend pos= outer north east]
\addplot coordinates {
(5, 0.886718) (20, 0.631851) (35, 0.460859) (50, 0.322725) (65, 0.211791) (80, 0.114317) (95, 0.026163) 
};
\addplot coordinates {
(5, 0.891100) (20, 0.642369) (35, 0.472371) (50, 0.333066) (65, 0.219830) (80, 0.119277) (95, 0.027443) 
};
\addplot coordinates {
(5, 0.934247) (20, 0.726504) (35, 0.554463) (50, 0.402130) (65, 0.271073) (80, 0.149124) (95, 0.035110) 
};

\end{axis}
\end{tikzpicture}
&
\begin{minipage}[b]{2.2cm}
\ifpdf\tikzexternaldisable\fi
\ref{convergelegend}\vspace*{0.85cm}
\ifpdf\tikzexternalenable\fi
\end{minipage}
\\
\begin{tikzpicture}[baseline, trim axis right]
\begin{axis}[width = 4.2cm, xtick = {5, 20, 35, 50, 65, 80, 95}, xmin = 1, ymin = 0, ymax = 1, ytick = {0, 0.25, 0.5, 0.75, 1}, cycle list name = yaf, ylabel = DNA,
yticklabels = {$0$, $.25$, $.5$, $.75$, $1$}, xlabel = asso]

\addplot coordinates {
(5, 0.644252) (20, 0.339485) (35, 0.173361) (50, 0.081682) (65, 0.031980) (80, 0.012126) (95, 0.001789) 
};
\addplot coordinates {
(5, 0.699359) (20, 0.383606) (35, 0.204455) (50, 0.093736) (65, 0.035329) (80, 0.013821) (95, 0.002882) 
};
\addplot coordinates {
(5, 0.734002) (20, 0.410947) (35, 0.223774) (50, 0.106046) (65, 0.040846) (80, 0.015742) (95, 0.003162) 
};
\pgfplotsextra{\yafdrawaxis{5}{95}{0}{1}}
\end{axis}
\end{tikzpicture}&
\begin{tikzpicture}[baseline, trim axis left, trim axis right]
\begin{axis}[width = 4.2cm, xtick = {5, 20, 35, 50, 65, 80, 95}, xmin = 1, ymin = 0, ymax = 1, ytick = {0, 0.25, 0.5, 0.75, 1}, ymajorticks = false, cycle list name = yaf, xlabel = krimp]
\addplot coordinates {
(5, 0.734109) (20, 0.519329) (35, 0.348619) (50, 0.254467) (65, 0.163454) (80, 0.102112) (95, 0.036762) 
};
\addplot coordinates {
(5, 0.755289) (20, 0.548472) (35, 0.376951) (50, 0.279122) (65, 0.182007) (80, 0.114949) (95, 0.041902) 
};
\addplot coordinates {
(5, 0.782367) (20, 0.578420) (35, 0.401917) (50, 0.299295) (65, 0.196550) (80, 0.127931) (95, 0.049402) 
};

\pgfplotsextra{\yafdrawxaxis{5}{95}}
\end{axis}
\end{tikzpicture}&
\begin{tikzpicture}[baseline, trim axis left, trim axis right]
\begin{axis}[width = 4.2cm, xtick = {5, 20, 35, 50, 65, 80, 95}, xmin = 1, ymin = 0, ymax = 1, ytick = {0, 0.25, 0.5, 0.75, 1}, ymajorticks = false, cycle list name = yaf, xlabel = itt]
\addplot coordinates {
(5, 0.737349) (20, 0.461814) (35, 0.313395) (50, 0.204605) (65, 0.121721) (80, 0.061256) (95, 0.013581) 
};
\addplot coordinates {
(5, 0.763815) (20, 0.500239) (35, 0.349466) (50, 0.233837) (65, 0.142121) (80, 0.072790) (95, 0.016391) 
};
\addplot coordinates {
(5, 0.792679) (20, 0.530384) (35, 0.367824) (50, 0.248117) (65, 0.149519) (80, 0.075265) (95, 0.016742) 
};

\pgfplotsextra{\yafdrawxaxis{5}{95}}
\end{axis}
\end{tikzpicture}
\end{tabular}%
\end{center}
\vspace{-1em}
\caption{Distance between top-$k$ tile sets and top-100 tile sets as a function of $k$. Rows represent datasets while the columns represent the methods.}
\label{fig:converge}
\end{figure}

First, we evaluate whether the measured distance converges to $0$ when two sets of tiles approximate each other with regard to background knowledge. 
To this end, we take the tile sets of \emph{asso}, \emph{krimp}, and \emph{itt}, as obtained on resp. the \emph{Abstracts} and \emph{DNA} datasets. In Figure~\ref{fig:converge} we plot, per method, the measured distance between the top-$k$ and top-$100$ tiles. We give the measurements for three different background knowledge settings, resp. no background knowledge, knowledge of the average density of the dataset, and the column margins. The tiles are sorted according to their output order, for \emph{asso} and \emph{itt}, and ascending on code length for \emph{krimp}.

As Figure~\ref{fig:converge} shows, measurements indeed converge to $0$ for higher $k$, i.e. when the two tile sets become more identical. Adding background information typically increases the distance. This is due to two reasons. First, when density is used, then we can infer additional differences
between the areas that are not covered by tiles, thus highlighting the differences. Second, when we are using column margins, we reduce $\kl{\ifam{M}}{\ifam{B}}$, the joint information w.r.t. the background knowledge, consequently increasing the distance. Interestingly, for \emph{Abstracts}, the distances for \emph{asso} in fact decrease when density is used as background knowledge. This is caused by the fact that \emph{asso} produces many overlapping tiles and these overlaps are emphasised when density is used.

\begin{figure}[t!]
\input{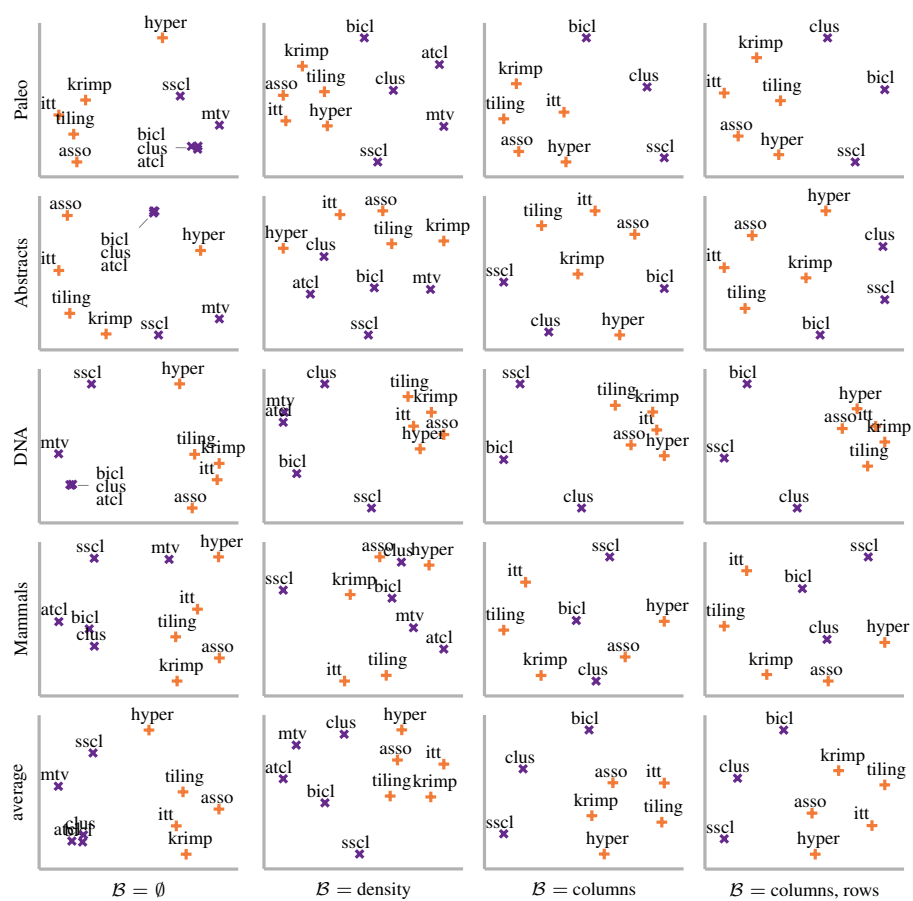}
\vspace{-1em}
\caption{Sammon projections (a type of multi-dimensional scaling, see~\cite{sammon:69:mapping}) of distances between tile sets. Each row represents a dataset and each column represents used background knowledge.\label{fig:mds}
Note that \emph{atcl} and \emph{mtv} are not included in the rightmost columns, as their tiles provide no information beyond column margins}
\vspace{-1em}
\end{figure}

\subsection{Distances between Results}

Our main experiment is to investigate how well we can compare between results of different methods. To do so, for every dataset, and every method considered, we convert their results into tile sets as described above. We measure the pair-wise difference between each of these tile sets, using resp. the empty set, overall density, the column margins, and the combination of row and column margins, as background knowledge. 
For analysis, we present these numerical results, and the averages over the datasets, visually in Figure~\ref{fig:mds} by plotting all pair-wise distances by Sammon projection~\cite{sammon:69:mapping}. Sammon projection is a commonly used variant of Multi-Dimensional Scaling (MDS), in which one projects a high-dimensional space (here the pair-wise distance tables) onto a two-dimensional space, while preserving the structure of the individual distances as well as possible. We colour tiling and mark and clustering methods differently. Note that by our conversion into noisy tiles, we here treat the results of \textsc{mtv} (itemsets and their frequencies) as attribute clusters.

Considering the first column first, we see that without background knowledge three of the clustering approaches provide virtually the same result. We also see that, albeit not identical, the results of \emph{asso}, \emph{itt}, \emph{krimp}, and \emph{tiling} are relatively close to each other; which makes sense from a conceptual point of view, as these methods are methodologically relatively similar. For \emph{DNA}, the measured dissimilarity between these methods lies between $0.28$ and $0.38$, whereas the dissimilarities to the other methods measure approximately $0.9$.

We observe that \emph{hyper}, while conceptually similar, is measured to provide different results when no background knowledge is given. This is mostly due to  it including a few very large tiles, that practically cover the whole data, whereas the other methods only cover the data partially. For \emph{hyper} we see that once background knowledge is included, these large tiles are explained away, and the method subsequently becomes part of the `tiling' group.

Clustering algorithms are close to each other when no background information is used because by covering all the data, they convey well that these datasets are sparse. When we use density as background knowledge, the differences between clusterings become visible.
Interestingly enough, adding row margins to column margins as background information has small impact on the distances.

\subsection{Redescribing Results}

\begin{table}[tb!]
\centering
\begin{tabular*}{\linewidth}{@{\extracolsep{\fill}}l c@{\hspace{1em}} r@{}l@{}r r@{}l@{}r r@{}l@{}r r@{}l@{}r r@{}l@{}r}
\toprule
&& \multicolumn{15}{l}{\textbf{Redescription / Full Tile Set (\# of Tiles)}}\\
\cmidrule{3-17}
\textbf{Dataset} && 
\multicolumn{3}{c}{\emph{asso}} & 
\multicolumn{3}{c}{\emph{tiling}} & 
\multicolumn{3}{c}{\emph{hyper}} & 
\multicolumn{3}{c}{\emph{itt}} & 
\multicolumn{3}{c}{\emph{krimp}} \\
\midrule
Abstracts &&  
$.76$ & $(70)$  & $.76$ & 
$.74$ & $(38)$  & $.74$ & 
$.50$ & $(32)$  & $.57$ & 
$.68$ & $(100)$ & $.68$ & 
$.85$ & $(98)$  & $.85$ \\ 
DNA &&  
$.65$ & $(11)$ & $.83$ & 
$.70$ & $ (9)$ & $.79$ & 
$.67$ & $(21)$ & $.84$ & 
$.67$ & $(11)$ & $.81$ & 
$.67$ & $(19)$ & $.81$\\ 
Mammals &&  
$.30$ & $(51)$ & $.31$ & 
$.68$ & $(3)$ & $.68$ & 
$.34$ & $(24)$ & $.39$ & 
$.78$ & $(2)$ & $.78$ & 
$.49$ & $(91)$ & $.49$\\ 
Paleo &&  
$.68$ & $(16)$ & $.83$ & 
$.69$ & $(28)$ & $.78$ & 
$.68$ & $(23)$ & $.81$ & 
$.73$ & $(21)$ & $.81$ & 
$.74$ & $(38)$ & $.79$\\ 
\bottomrule
\end{tabular*}
\caption{Redescribing the results of clustering (\emph{clus}) by the results of resp. \emph{asso}, \emph{tiling}, \emph{hyper}, \emph{itt}, and \emph{krimp}, using the \textsc{Fruits} algorithm. Shown are, per combination, the dissimilarity between \emph{clus} and the discovered redescription, the number of tiles of the redescription, and the dissimilarity measurement to the  complete tile set.}\label{tbl:redesc} 
\end{table}

\begin{figure}[tb!]
\begin{center}
\begin{tikzpicture}[font = \small, draw = black!30, line width = 1.2pt]
\node [text width=3.4cm, rounded corners, draw] (krimp)
{associ rule\\
significantli outperform\\
high dimension\\
experiment evalu show\\
vector support machin};
\node [above =0.15cm of krimp.north west, anchor = base west] {\emph{krimp}};

\node [text width=3.3cm, rounded corners, draw , right =0.3cm of krimp.north east, anchor = north west] (tijl)
{vector support machin\\
associ rule \\
dimension\\
outperform};
\node [above =0.15cm of tijl.north west, anchor = base west] {\emph{itt}, $d = 0.77$};

\node [text width=4.1cm, rounded corners, draw, right =0.3cm of tijl.north east, anchor = north west] (asso)
{
associ rule mine algo\\
vector method support\\
algo method high dimension\\
algo show\\
};
\node [above =0.15cm of asso.north west, anchor = base west] {\emph{asso}, $d = 0.83$};

\end{tikzpicture}
\end{center}
\vspace{-1em}
\caption{Redescribing, using the \textsc{Fruits} algorithm, (left) 5 tiles selected from the \emph{krimp} result on the \emph{Abstracts} dataset, by resp. tiles from (centre) \emph{itt}  and (right) \emph{asso}.}\label{fig:redisc}
\end{figure}

Next, we empirically evaluate how our measure can be employed with regard to  redescribing  results; both as validation as well as possible application. 

To this end, we first investigate redescribing results of completely different methods. As such, we take clustering as the target and density as the background information, and redescribe its result by using the tile sets of five of the pattern mining methods as candidate tile sets.
Table~\ref{tbl:redesc} shows the results of these experiments: for four datasets, the measured divergence between the redescription and the target, the divergence of the complete tile set to the target, and the number of tiles selected for the redescription. First, and foremost, we see that redescription decreases the measured divergence, which correctly shows that by filtering out, e.g. too specific, tiles that provide information not in the target, we obtain a better description of the target.

We also see, with the exception of \emph{Mammals}, that the measurements are quite high overall, suggesting the clustering provides information these results do not; not surprising, as these pattern mining methods focus on covering $1$s, and not necessarily cover the whole data. Indeed, we see that by providing large and noisy tiles, and so covering more of the data, \emph{asso} and \emph{hyper} lead to the best redescriptions of the clustering results. In particular for \emph{Mammals}, the target can be approximated very well, by resp. only half and a quarter of the total tiles. Overall, we note that typically only fractions of the full tile sets are selected, yet the amount of shared information is larger than for the full tile set: the pattern mining methods provide detailed local information not captured by clustering.

Second, we take a closer look at individual redescriptions. In order to be able to interpret these, we use the \emph{Abstracts} dataset. We use \emph{asso}, \emph{krimp}, and \emph{itt}, as these provide sufficiently many tiles to choose from; we leave \emph{hyper} out, as for this data it mostly gives only very general tiles, covering all $1$s in only $100$ tiles.

By hand, we select $5$ out of $100$ \emph{krimp} tiles, and we identify, for \emph{asso} and \emph{itt}, the sets of tiles that best approximate that partial result, and investigate how well the target concepts are approximated. In Figure~\ref{fig:redisc}, we give an example. By the high distances, $0.77$ and $0.83$, we see the target is not approximated in detail. Overall, for \emph{itt} we find only high-level translations that leave out detail, as its full tile set consists mostly of small itemsets. For \emph{asso}, we see the redescription consists of target tiles combined with general concepts, which together give a reasonable approximation of the target. It is important to note that not simply all intersecting itemsets are given, but only those that provide sufficient information on the target; for both methods, overly large and overly general tiles (e.g. `high') are not included in the redescription.

\subsection{Iterative Data Mining}

Next, we empirically evaluate how our measure can be employed for iterative data mining; iteratively identifying that (partial) result that provides most novel information with regard to the data.

First, for the \textit{DNA} dataset, we combine the results of \textit{hyper, tiling, krimp, itt, asso}, and \textit{sscl} into one target tile set. Next, we apply the \textsc{Fitamin} algorithm to rank these tiles by iteratively finding the tile that makes for the most informative addition to our selection. We do this for our four settings of background knowledge, and visually depict the results in Figure~\ref{fig:idm:dna}.

\begin{figure}
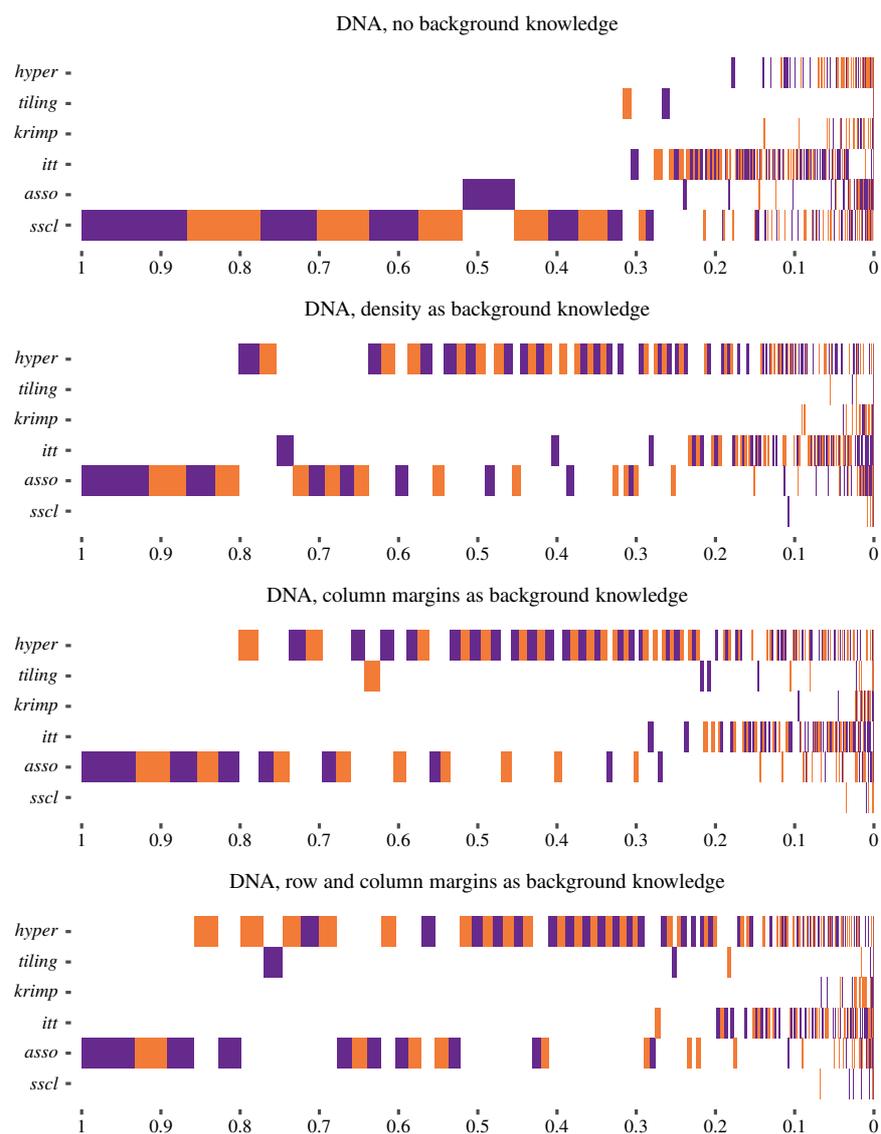

\begin{tikzpicture}
\begin{axis}[ymin=0, xmin=0, ymax=6,xmax=1,grid=none,height=4cm, width = 12cm, title = {DNA, no background knowledge},
	ytick={0.5,...,5.5}, 
	yticklabels={\emph{sscl}, \emph{asso}, \emph{itt}, \emph{krimp}, \emph{tiling}, \emph{hyper}},
	xtick = {0, 0.1, ...,  1.1}, xticklabels={1, 0.9, 0.8, 0.7, 0.6, 0.5, 0.4, 0.3, 0.2, 0.1, 0}]
\input{figs/dna-iter-none-heur.tex}
\end{axis}
\end{tikzpicture}

\vspace{0.5em}
\begin{tikzpicture}
\begin{axis}[ymin=0, xmin=0, ymax=6,xmax=1,grid=none,height=4cm, width = 12cm,
title = {DNA, density as background knowledge},
	ytick={0.5,...,5.5}, 
	yticklabels={\emph{sscl}, \emph{asso}, \emph{itt}, \emph{krimp}, \emph{tiling}, \emph{hyper}},
	xtick = {0, 0.1, ...,  1.1}, xticklabels={1, 0.9, 0.8, 0.7, 0.6, 0.5, 0.4, 0.3, 0.2, 0.1, 0}]
\input{figs/dna-iter-dens-heur.tex}
\end{axis}
\end{tikzpicture}

\vspace{0.5em}
\begin{tikzpicture}
\begin{axis}[ymin=0, xmin=0, ymax=6,xmax=1,grid=none,height=4cm, width = 12cm,
title = {DNA, column margins as background knowledge},
	ytick={0.5,...,5.5}, 
	yticklabels={\emph{sscl}, \emph{asso}, \emph{itt}, \emph{krimp}, \emph{tiling}, \emph{hyper}},
	xtick = {0, 0.1, ...,  1.1}, xticklabels={1, 0.9, 0.8, 0.7, 0.6, 0.5, 0.4, 0.3, 0.2, 0.1, 0}]
\input{figs/dna-iter-cm-heur.tex}
\end{axis}
\end{tikzpicture}

\vspace{0.5em}
\begin{tikzpicture}
\begin{axis}[ymin=0, xmin=0, ymax=6,xmax=1,grid=none,height=4cm, width = 12cm,
title = {DNA, row and column margins as background knowledge},
	ytick={0.5,...,5.5}, 
	yticklabels={\emph{sscl}, \emph{asso}, \emph{itt}, \emph{krimp}, \emph{tiling}, \emph{hyper}},
	xtick = {0, 0.1, ...,  1.1}, xticklabels={1, 0.9, 0.8, 0.7, 0.6, 0.5, 0.4, 0.3, 0.2, 0.1, 0}]
\input{figs/dna-iter-cmrm-heur.tex}
\end{axis}
\end{tikzpicture}

\caption{Selection order (from left to right), and gain in information (bar
width) with regard to the target of the combined results of \textit{hyper,
tiling, krimp, itt, asso}, and \textit{sscl}, for the \textit{DNA}. Top to
bottom for background knowledge consisting resp. of the empty set, the data
density, the column margins, and both column and row
margins. The $x$-axis is the distance between the selected tiles so far and
the full tile collection.}\label{fig:idm:dna}
\end{figure}

We see that if we know nothing about the data, the \textit{sscl} tiles, which cover relatively large areas of $0$s, provide much information; whereas the tiles of the other methods, depicting mainly structure of the $1$s, are ranked much lower. Once the density of the data is included in $\ifam{B}$, we already know the data is relatively sparse---and subsequently tiles that identify relatively large areas of relatively many ones are initially most informative. 
For all non-empty background knowledge, we see the tiles of \textit{tiling}, \textit{krimp}, \textit{itt}, and \textit{sscl} to be presented quite late in the process. This makes sense, as these methods are aimed at identifying interesting local structure, which in an iterative setting is only informative once we know the bigger picture.

To reduce computation, \textsc{Fitamin} employs a heuristic to determine which tile to add next to its selection---as opposed to iteratively fitting maximum entropy models for every candidate, and measuring differences between those. To investigate the quality of this heuristic, we ran a variant of \textsc{Fitamin} that does calculate exact distances for the \textit{Paleo} dataset. While this variant is able to optimally locally select the best addition, we only observe very slight differences in the final rankings; the main differences occur when we have the choice between two or more tiles approximately providing the same gain in information, and the tie is broken differently between the exact and heuristic approach. When overall we take a look at the individual differences between estimates and real distances, we see the estimates are of very high quality; the largest absolute recorded difference we recorded was only $0.007$ (on a range of $0$ to $1$), which is close to numerical stability.

\begin{figure}
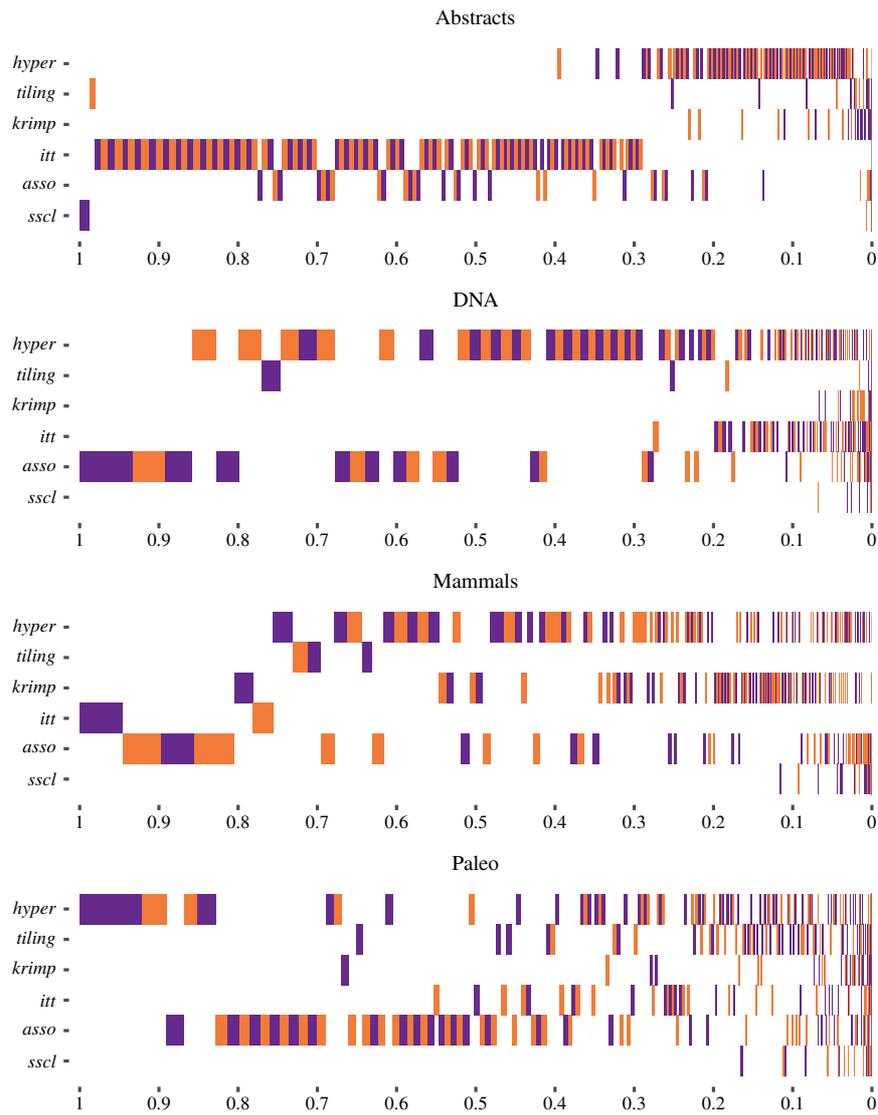

\begin{tikzpicture}
\begin{axis}[ymin=0, xmin=0, ymax=6,xmax=1,grid=none,height=4cm, width = 12cm, title = {Abstracts},
	ytick={0.5,...,5.5}, 
	yticklabels={\emph{sscl}, \emph{asso}, \emph{itt}, \emph{krimp}, \emph{tiling}, \emph{hyper}},
	xtick = {0, 0.1, ...,  1.1}, xticklabels={1, 0.9, 0.8, 0.7, 0.6, 0.5, 0.4, 0.3, 0.2, 0.1, 0}]
]
\input{figs/icdm-iter-cmrm-heur.tex}
\end{axis}
\end{tikzpicture}

\vspace{0.5em}
\begin{tikzpicture}
\begin{axis}[ymin=0, xmin=0, ymax=6,xmax=1,grid=none,height=4cm, width = 12cm,
title = {DNA},	ytick={0.5,...,5.5}, 
	yticklabels={\emph{sscl}, \emph{asso}, \emph{itt}, \emph{krimp}, \emph{tiling}, \emph{hyper}},
	xtick = {0, 0.1, ...,  1.1}, xticklabels={1, 0.9, 0.8, 0.7, 0.6, 0.5, 0.4, 0.3, 0.2, 0.1, 0}]
\input{figs/dna-iter-cmrm-heur.tex}
\end{axis}
\end{tikzpicture}

\vspace{0.5em}
\begin{tikzpicture}
\begin{axis}[ymin=0, xmin=0, ymax=6,xmax=1,grid=none,height=4cm, width = 12cm,
title = {Mammals},	ytick={0.5,...,5.5}, 
	yticklabels={\emph{sscl}, \emph{asso}, \emph{itt}, \emph{krimp}, \emph{tiling}, \emph{hyper}},
	xtick = {0, 0.1, ...,  1.1}, xticklabels={1, 0.9, 0.8, 0.7, 0.6, 0.5, 0.4, 0.3, 0.2, 0.1, 0}]
\input{figs/mammals-iter-cmrm-heur.tex}
\end{axis}
\end{tikzpicture}

\vspace{0.5em}
\begin{tikzpicture}
\begin{axis}[ymin=0, xmin=0, ymax=6,xmax=1,grid=none,height=4cm, width = 12cm,
title = {Paleo},	ytick={0.5,...,5.5}, 
	yticklabels={\emph{sscl}, \emph{asso}, \emph{itt}, \emph{krimp}, \emph{tiling}, \emph{hyper}},
	xtick = {0, 0.1, ...,  1.1}, xticklabels={1, 0.9, 0.8, 0.7, 0.6, 0.5, 0.4, 0.3, 0.2, 0.1, 0}]
\input{figs/paleo-iter-cmrm-heur.tex}
\end{axis}
\end{tikzpicture}

\caption{Selection order (from left to right), and gain in informativeness
(width of the bars) with regard to the target tile set of the combined
results of \textit{hyper, tiling, krimp, itt, asso}, and \textit{sscl},
for, from top to bottom, \textit{Abstracts}, \textit{DNA}, \textit{Mammals},
and \textit{Paleo}. The $x$-axis is the distance between the selected tiles so
far and the full tile collection.}\label{fig:idm:alldata}

\end{figure}

Next, for all datasets, using row and column margins as background knowledge, we show the selection order and gains in information as Figure~\ref{fig:idm:alldata}. As one would expect for datasets with different characteristics, different methods are best suited to extract the most information---there is no free lunch. Overall, and as expected, we observe that initially relatively large tiles not already explained away by the background knowledge are selected. As such, tiles from \textit{asso} and \textit{hyper} are often selected first, as these give a broad-strokes view of the data. 
Tiles that give much more local detail, or focus only on the $1$s of the data, such as those mined by \textit{krimp} and \textit{itt} are most informative later on.
Do note that \textsc{Fitamin} does not give a qualitative ranking of results: our measure says nothing about interpretability, only about relative informativeness.

\section{Discussion}\label{sec:disc}

The goal of this paper is to take a first step towards comparing between the results of different data mining methods. The method we propose here is for results obtained on binary data. By developing further maximum entropy modelling techniques, however, the same basic idea could be applied to richer data types. 
A recent result by~\cite{konto:11:icdm} formalises a maximum entropy model for real-valued data. If this model can be extended to handle real-valued tiles,
then our approach is also applicable for results on real-valued data. As for real-valued data statistics much richer than an average over a tile can be identified, a future work will involve in defining sufficiently rich theory formalising maximum entropy models that can take such statistics within tiles, such as correlations between attributes, into account as background information.

Besides measuring divergence of results between different methods, using the \textsc{Fruits} algorithm we can identify high-quality redescriptions of results, and by \textsc{Fitamin} we can heuristically approximate the optimal ranking of partial results based on their iterative informativeness. While beyond the scope of this paper, the latter algorithm can also be used rank complete results, i.e. tile sets. Moreover, our measure could be extended for choose the most informative result out of many runs of a randomised method, or, to measure differences between results when varying parameters of a method. Another, likely more challenging, extension would be to identify a `centroid' or `medioid' result that summarises a given large collection of results well.

Currently we rather straightforwardly convert results into tile sets. While we so capture much of the information they provide, we might not capture all information present in a result. 
Ideally, we would be able to encode structure beyond simple tile densities, such that the information captured in a result can be maintained more precisely when converted into sets of tiles. 
A desirable extension would be to be able to specify for a tile which attribute-value combinations occur how often within it, which would allow for better conversion of the results of Attribute Clustering~\citep{mampaey:10:summarising}, as well as allow us to take simple itemset frequencies into account. A second refinement for future work would be to be able to specify complex structures in the data captured by statistics such as Lazarus counts~\citep{fortelius:06:spectral}, or nestedness~\citep{mannila:07:nestedness}. 
Such refinements require further theory on maximum entropy modelling. More complex modelling aside, our general approach of comparing results by comparing how many possible datasets exhibit the identified structure remains the same.

Our distance is based on maximum entropy model which uses the available
information --- that is the given tiles --- as efficiently as possible.
For example, if the frequencies of tiles in the second set can be derived
from the first set, then the distance will be $0$ between these two sets.
However, we do not take into account how complex are the derivations for
the human and this might lead to some counterintuitive results.
Another restriction of the maximum entropy model we currently employ, is that it cannot recognise the informativeness of a data mining result
identifying that the data is uniformly distributed. While this can be a highly informative, and powerful result, as our maximum entropy model 
essentially is the generalised uniform distribution under constraints, the uniform property follows naturally, is no deviation, and hence is not deemed informative. 
These restrictions are both due to our choice of modelling the space of all datasets by maximum entropy. While a well-founded choice, with many desirable
properties, it is a choice nevertheless and future work may identify other
probabilistic models that circumvent the above mentioned restrictions.

It is very important to note that our method is \emph{not} a quality measure on data mining results. It solely measures the information shared between two sets of tiles; it does not measure the subjective quality of results, nor does it say anything about the ease of analysis of a result. Note that this is a good thing, and our explicit goal. Ease of analysis is highly subjective, being dependent on the expert and the available means. By our measure, the expert can check how informative different results are, and given this information decide which result to analyse in detail. If two results of different complexity of analysis  are measured to be approximately equally informative, the logical decision would be to analyse the more simple one; only investing the time and effort of considering the more complex result in detail if the measurement shows it will likely provide a lot of information not also available in the simple result.

\section{Conclusion}\label{sec:concl}

In this paper we discussed comparing results of different explorative data
mining algorithms. We argued that any mining result identifies some properties
of the data, and that, in an abstract way, we can identify all datasets for
which these properties hold. By incorporating these properties into a model using the
Maximum Entropy principle, we can measure the shared amount of
information by Kullback-Leibler divergence.  The measure we construct this way is
flexible, and naturally allows including background knowledge, such that
differences in results can be measured from the perspective of what a user
already knows. 

As a first step towards comparing results in general, we formalised our approach for binary data, showed that we can convert results into tiles, and discussed how to incorporate these into a Maximum Entropy model. Our approach provides a means to study and tell differences between results of different data mining methods. For applying the measure, we gave the \textsc{Fruits} algorithm for parameter-freely identifying which parts of a given set of results best redescribe a given (partial) result, and the \textsc{Fitamin} algorithm for iteratively finding that result that provides maximal novel information with regard to what we have learned so far. Experiments  showed our measure gives meaningful results, correctly identifies methods that are similar in nature, automatically identifies sound redescriptions of results, and is highly applicable for iterative data mining.

\section*{Acknowledgements}
The authors wish to thank, in alphabetical order: 
Tijl De Bie for his information-theoretic noisy tile miner~\citep{debie:11:dami}; 
David Fuhry for his implementation of \textsc{Hyper}~\citep{xiang:10:hyper}; 
Matthijs van Leeuwen for major contributions to the implementation of \textsc{Krimp}~\citep{vreeken:11:krimp}; 
Stijn Ligot for running experiments with \textsc{ProClus}~\citep{aggarwal:99:proclus,muller:09:sscl}; 
Michael Mampaey for the implementations of Attribute Clustering~\citep{mampaey:10:summarising,mampaey:12:summarising} as well as \textsc{mtv}~\citep{mampaey:12:tkdd}; 
Pauli Miettinen for his implementation of \textsc{Asso}~\citep{miettinen:08:positive}; 

Nikolaj Tatti and Jilles Vreeken are both supported by a Post-Doctoral Fellowship of the Research Foundation -- Flanders (\textsc{fwo}).

\bibliographystyle{spbasic}      
\bibliography{bib/abbreviations,bib/bib-jilles,bib/bibliography}   

\appendix
\section{Proofs for Theorems}
\label{sec:apx}

\begin{proof}[of Theorem~\ref{thr:exponential}]
Let $\pemp$ be the maximum entropy distribution. Define a distribution $q$
\[
	q(D) = \prod_{i, j} \pemp((i, j) = D(i, j)) \quad .
\]
Lemma~\ref{lem:alternative} now implies that $q \in \mathcal{P}$
so Theorem~\ref{thr:maxent} implies that $\set{D \mid \pemp(D) = 0} \subseteq \set{D \mid q(D) = 0}$.
Assume that
$\pemp(D) > 0$. Then $\pemp((i, j) = D(i, j)) > 0$ and consequently $q(D) > 0$.
This implies that $\set{D \mid \pemp(D) = 0} = \set{D \mid q(D) = 0}$.

Assume that $\pemp(D) > 0$, then we can decompose the sums of $\freq{T; D}$ in
the exponential form in Theorem~\ref{thr:maxent} $\pemp(D)$ can be written
as a product of potentials, $\pemp(D) \propto \prod_{i, j}\phi_{i, j}(D(i, j))$. If we normalize each
term by $\phi_{i, j}(1) + \phi_{i, j}(0)$, we have that 
$\pemp(D) = \prod_{i, j} \pemp((i, j) = D(i, j))$ for $\pemp(D) > 0$.
If on the other hand $\pemp{D} = 0$, then also $q(D) = 0$, so by definition
$\pemp((a, b) = D(a, b)) = 0$ for some $(a, b)$. This implies that $\pemp(D) = 0 = \prod_{i, j} \pemp((i, j) = D(i, j))$.

Let us show that the individual terms have the stated form.
Fix $(i, j)$ such that $0 < \pemp((i, j) = 1) < 1$. Let $D$ be a dataset and let $D'$
be the dataset with flipped $(i, j)$th entry. Then $q(D) > 0$ if and only if $q(D') > 0$
and hence $\pemp(D) > 0$ if and only if $\pemp(D') > 0$.
Divide the dataset space
\[
	\dspace_v = \set{D \in \dspace \mid D(i, j) = v, \pemp(D) > 0} \quad .
\]
The previous argument shows that for each $D \in \dspace_1$, the dataset $D'$ with flipped entry $D' \in \dspace_0$.
Theorem~\ref{thr:maxent} implies that $\pemp(D)/\pemp(D') = \exp\fpr{\sum_{T \in \ifam{T}(i, j)} \lambda_T}$.
We have now
\[
	\frac{\pemp((i, j) = 1)}{\pemp((i, j) = 0)} = \frac{\sum_{D \in \dspace_1} \pemp(D)}{\sum_{D \in \dspace_0} \pemp(D)}  = \frac{\sum_{D \in \dspace_1} \pemp(D)}{\sum_{D \in \dspace_1} \pemp(D')}  = \exp \sum_{T \in \ifam{T}(i, j)} \lambda_T \quad .
\]
The theorem follows by rearranging the terms.\qed
\end{proof}

\begin{proof}[of Theorem~\ref{thr:exactfactor}]
We will prove the theorem by showing that $\pemp_\ifam{T}$ satisfies the
conditions in Theorem~\ref{thr:exponential}.  Let $\mathcal{P}$ be the
distributions satisfying the frequencies.  Let $\mathcal{Z}$ be the collection
of datasets for which $p(D) = 0$ for every $p \in \mathcal{P}$. Set $\mathcal{Z}' =
\set{D \mid \pemp_\ifam{T}(D) = 0}$.
Let $q \in \mathcal{P}$ and let $(i, j) \in \area{T}$.
Corollary~\ref{cor:exact} implies that $q((i, j) = 1) = \alpha_T$.  This
implies that $\mathcal{Z}' \subseteq Z$ and since $\pemp_\ifam{T} \in
\mathcal{P}$, we must have $\mathcal{Z} = \mathcal{Z}'$. By setting $\lambda_T
= 0$ for every tile we see that $\pemp_\ifam{T}$ has the exponential form given
in Theorem~\ref{thr:exponential}. Hence, $\pemp_\ifam{T}$ is truly
the maximum entropy distribution.
\qed
\end{proof}

Theorem~\ref{thr:distanceexact} follows directly from the following lemma.
\begin{lemma}
\label{lem:klexact}
Let $\ifam{T}$ and $\ifam{U}$ be two tile sets such that $\area{\ifam{U}} \subseteq \area{\ifam{T}}$. Then
$\kl{\ifam{T}}{\ifam{U}} = \abs{\area{\ifam{T}} \setminus \area{\ifam{U}}} \log 2$.
\end{lemma}

\begin{proof}
Let $X = \area{\ifam{T}} \setminus \area{\ifam{U}}$.  Theorem~\ref{thr:exponential} implies that
\[
    \kl{\ifam{T}}{\ifam{U}} = \sum_{i, j} \sum_{v = 0, 1} \pemp_\ifam{T}((i, j) = v) \log \frac{\pemp_\ifam{T}((i, j) = v)}{\pemp_\ifam{U}((i, j) = v)} \quad .
\]
The only non-zero terms in the sum are the entries $(i, j) \in X$ for which $\pemp_\ifam{T}((i, j) = v) = 1$.
Hence we have
\[
	\kl{\ifam{T}}{\ifam{U}} = \sum_{(i, j) \in X} 1 \times \log \frac{1}{1/2} = \abs{X}\log 2 \quad .
\]
This proves the lemma.\qed
\end{proof}

To prove the next theorems we will need the following result.

\begin{lemma}
Let $\ifam{T}$ be a set of tiles. Let $\pemp$ be the corresponding maximum
entropy model. Let $q$ be a distribution such that $\freq{T; q} = \freq{T; \pemp}$ for $T \in \ifam{T}$.
Then
\[
	\sum_{D \in \dspace} q(D) \log \pemp(D) = \lambda_0 +  \sum_{T \in \ifam{T}} \lambda_{T}\abs{\area{T}}\freq{T; \pemp} = -\ent{\pemp},
\]
where $\lambda_T$ are the weights as defined in Theorem~\ref{thr:maxent}, and $\lambda_0$
is the normalisation constant of $\pemp$.
\end{lemma}

\begin{proof}
Let $\ifam{Z} = \set{D \in \dspace \mid \pemp(D) = 0}$. Note that $q(D) = 0$ for any $D \in \ifam{Z}$.
Applying Theorem~\ref{thr:maxent} leads us to
\[
\begin{split}
	\sum_{D \in \dspace} q(D) \log \pemp(D) & = \sum_{D \in \dspace \setminus \ifam{Z}} q(D) \log \pemp(D) \\
	                                        & = \sum_{D \in \dspace \setminus \ifam{Z}} q(D) \big( \lambda_0 + \sum_{T \in \ifam{T}} \lambda_{T}\abs{\area{T}}\freq{T; D}\big) \\
	                                        & = \lambda_0 + \sum_{T \in \ifam{T}} \lambda_T\abs{\area{T}} \sum_{D \in \dspace \setminus \ifam{Z}} q(D) \freq{T; D} \\
	                                        & = \lambda_0 + \sum_{T \in \ifam{T}} \lambda_T\abs{\area{T}} \freq{T; \pemp}\quad. \\
\end{split}
\]
This proves the left side of the equation.
To prove the right side apply the same argument with $q = \pemp$.
\qed
\end{proof}

\begin{corollary}
\label{cor:diff}
Let $\ifam{A}$ and $\ifam{B}$ be two tile collections such that $\ifam{B} \subseteq \ifam{A}$. Then
	$\kl{\ifam{A}}{\ifam{B}} = \ent{\ifam{B}} - \ent{\ifam{A}}$.
\end{corollary}
\begin{proof} Let $\pemp_\ifam{A}$ and $\pemp_\ifam{B}$ be two maximum entropy distributions.
Since $\freq{T; \pemp_\ifam{A}} = \freq{T; \pemp_\ifam{B}}$  for any $T \in \ifam{B}$,
we have
\[
	\kl{\ifam{A}}{\ifam{B}} = -\sum_{D \in \dspace} \pemp_\ifam{A}(D) \log \pemp_\ifam{B}(D) - \ent{\ifam{A}} = \ent{\ifam{B}} - \ent{\ifam{A}},
\]
which proves the result.
\qed
\end{proof}

\begin{proof}[of Theorem~\ref{thr:bound}]
The case for exact tiles follows directly from Theorem~\ref{thr:distanceexact}. The general case
follows from Corollary~\ref{cor:diff} by first noticing that
\[
	\kl{\ifam{M}}{\ifam{U} \cup \ifam{B}} = \ent{\ifam{U} \cup \ifam{B}} - \ent{\ifam{M}} \leq \ent{\ifam{B}} - \ent{\ifam{M}} = \kl{\ifam{M}}{\ifam{B}}
\]
and similarly for $\kl{\ifam{M}}{\ifam{T} \cup \ifam{B}}$. Thus we have
\[
	\dist{\ifam{T}, \ifam{U} ; \ifam{B}} = \frac{\kl{\ifam{M}}{\ifam{U} \cup \ifam{B}} + \kl{\ifam{M}}{\ifam{T} \cup \ifam{B}}}{\kl{\ifam{M}}{\ifam{B}}} \leq 2\frac{\kl{\ifam{M}}{\ifam{B}}}{\kl{\ifam{M}}{\ifam{B}}} = 2,
\]
which proves the result.
\qed
\end{proof}

\begin{proof}[of Theorem~\ref{thr:indcol}]
Define $\ifam{X} = \ifam{T} \cup \ifam{B}$ and $\ifam{Y} = \ifam{U} \cup
\ifam{B}$.  Let $\pemp_\ifam{X}$, $\pemp_\ifam{Y}$, $\pemp_\ifam{B}$, and
$\pemp_\ifam{M}$ be the corresponding maximum entropy distributions.
Define a distribution $q(D) \propto \pemp_\ifam{X}(D)\pemp_\ifam{Y}(D)$.

We claim that $q = \pemp_\ifam{M}$. To prove this we first show that $q$
satisfies the tiles in $\ifam{M}$. Let $e = (i, j)$ be an entry.
Assume that $e \in \area{\ifam{B}}$. Since the tiles in $\ifam{B}$ are exact,
we have, due Corollary~\ref{cor:exact},
\[
	\pemp_\ifam{X}(e = 1) = \pemp_\ifam{Y}(e = 1) = \pemp_\ifam{M}(e = 1) = v, \quad\text{ where }\quad v = 0,1\quad.
\]
Assume that $v  = 1$. This means that $\pemp_\ifam{Y}(D) = 0$, and consequently
$q(D) = 0$, if $D(e) = 0$.  This implies that $q(e = 1) = 1$. Similar argument
holds for $v = 0$. Hence, $q$ satisfies tiles from $\ifam{B}$.

Assume now that $e \notin \area{\ifam{B}}$.  Since $q$ has an exponential form,
$q(e = 1)$ can be written in the form given in Theorem~\ref{thr:exponential}.
Assume that $e \in \area{\ifam{T}}$, then $e \notin \area{\ifam{U}}$. Assume that $0 <
q(e = 1) < 1$.  This implies that the fraction given in
Theorem~\ref{thr:exponential} contains only weights  from the tiles in
$\ifam{T}$.  Hence, $\pemp_\ifam{X}(e = 1) = q(e = 1)$. If $q(e = 1) = 0, 1$,
then either $\pemp_\ifam{X}(e = 1) = 0, 1$ or $\pemp_\ifam{Y}(e = 1) = 0, 1$.
Since $e \notin \area{\ifam{Y}}$, we must have $\pemp_\ifam{Y}(e = 1) = 1/2$.
Hence, $\pemp_\ifam{X}(e = 1) = q(e = 1)$.  This implies that $q$ satisfies
tiles in $\ifam{T}$. Similarly, $q$ satisfies the tiles in $\ifam{U}$.
Hence $q$ satisfies the tiles in $\ifam{M}$.

Theorem~\ref{thr:maxent} now implies that $\set{D \in \dspace \mid
\pemp_\ifam{M}(D) = 0} \subseteq \set{D \in \dspace \mid q(D) = 0}$.  To prove
the other direction, let $D$ be such that $q(D) = 0$. Then $\pemp_\ifam{X}(D) =
0$ or $\pemp_\ifam{Y}(D) = 0$. Assume that $\pemp_\ifam{X}(D) = 0$.
Theorem~\ref{thr:maxent} implies that any distribution that satisfies
$\ifam{X}$ must have vanish at $D$.  Hence, $\pemp_\ifam{M}(D) = 0$.
Consequently, $\set{D \in \dspace \mid \pemp_\ifam{M}(D) = 0} = \set{D \in \dspace \mid q(D) = 0}$

Since $q$ has the correct exponential form and satisfies the tiles in
$\ifam{M}$, Theorem~\ref{thr:maxent} implies that $q = \pemp_\ifam{M}$.

We have the following properties
\[
\begin{split}
	e \in \area{\ifam{B}} & \implies q(e = 1) = \pemp_\ifam{B}(e = 1) = \pemp_\ifam{X}(e = 1) = \pemp_\ifam{Y}(e = 1), \\
	e \in \area{\ifam{T}} & \implies q(e = 1) = \pemp_\ifam{X}(e = 1) \quad\text{ and }\quad \pemp_\ifam{Y}(e = 1) = \pemp_\ifam{B}(e = 1), \\
	e \in \area{\ifam{U}} & \implies q(e = 1) = \pemp_\ifam{Y}(e = 1) \quad\text{ and }\quad \pemp_\ifam{X}(e = 1) = \pemp_\ifam{B}(e = 1), \\
	e \notin \area{\ifam{M}} & \implies q(e = 1) = \pemp_\ifam{B}(e = 1) = \pemp_\ifam{X}(e = 1) = \pemp_\ifam{Y}(e = 1)\quad. \\
\end{split}
\]
Theorem~\ref{thr:compute} now implies that
\[
	\kl{\ifam{M}}{\ifam{B}}	 = \kl{\ifam{X}}{\ifam{B}} + \kl{\ifam{Y}}{\ifam{B}}
\]
and
\[
	\kl{\ifam{M}}{\ifam{X}}	 = \kl{\ifam{Y}}{\ifam{B}} \quad\text{ and }\quad
	\kl{\ifam{M}}{\ifam{Y}}	 = \kl{\ifam{X}}{\ifam{B}},
\]
which proves the theorem.
\qed
\end{proof}

\begin{proof}[of Theorem~\ref{thr:triangle}]
Theorem~\ref{thr:distanceexact} implies that we can prove the claim by showing that
\[
	1 - \frac{\abs{A \cap B}}{\abs{A \cup B}} \leq 1 - \frac{\abs{A \cap C}}{\abs{A \cup C}} + 1 - \frac{\abs{B \cap C}}{\abs{B \cup C}},
\]
where $A$, $B$, and $C$ are finite sets. We can rewrite the inequality as
\begin{equation}
\label{eq:normineq}
	\frac{\abs{A \cap C}}{\abs{A \cup C}} + \frac{\abs{B \cap C}}{\abs{B \cup C}} - \frac{\abs{A \cap B}}{\abs{A \cup B}} \leq 1\quad. 
\end{equation}
To prove the inequality we will modify the sets in such a way that it will only
increase the left side. 

Assume that there exists $x \in C \setminus (A \cup B)$. Removing $x$ from $C$
will decrease the denominators of the first two terms by $1$, and thus increase
the left side of the inequality. Hence we can assume safely that $C \subseteq A
\cup B$.

Assume that there exists $x \in (A \cap B) \setminus C$. Adding $x$ to $C$ will
increase the numerators of the first two terms by $1$, and thus increase the
left side of the inequality. Hence we can assume safely that $A \cap B
\subseteq C$.

Now let us assume that we have three variables $x \in A \cap B$, $y \in A
\setminus C$, and $z \in B \setminus C$. If we remove $x$ from $A$ and $B$ and
$C$, add $z$ to $C$, and add $y$ to $C$, then the numerators and the
denominators of the first two terms will increase by $1$ and the numerator and
denominator of the third term will decrease by $1$. Since
\[
	\frac{x}{y} \leq \frac{x + 1}{y + 1},
\]
for any $0 \leq x \leq y \neq 0$, it follows that the left side of
Eq.~\ref{eq:normineq} can only increase.

We repeat this replacement until two cases happen (this is possible since $A$,
$B$, and $C$ are finite). Either $A \cap B = \emptyset$ or $A \subseteq C$ (or
$B \subseteq C$).

Assume that $A \cap B = \emptyset$. Then Eq.~\ref{eq:normineq} reduces to
\[
    \frac{\abs{A \cap C}}{\abs{A \cup C}} + \frac{\abs{B \cap C}}{\abs{B \cup C}} \leq 1\quad. 
\]
Assume that there exists $x \in A \setminus C$. Then we can remove $x$ from $A$
and only increase the first term. Hence we can assume that $A \subseteq C$ and
similarly $B \subseteq C$. This implies that $C = A \cup B$, and we get
\[
    \frac{\abs{A \cap C}}{\abs{A \cup C}} + \frac{\abs{B \cap C}}{\abs{B \cup C}}
	= \frac{\abs{A}}{\abs{C}} + \frac{\abs{B}}{\abs{C}} 
	= \frac{\abs{A}}{\abs{A} + \abs{B}} + \frac{\abs{B}}{\abs{A} + \abs{B}} = 1,
\]
which proves the first case.

To prove the second case, assume that $A \subseteq C$.  This implies that $A
\cup B \subseteq C \cup B \subset A \cup B \cup B = A \cup B$, that is $A \cup
B = B \cup C$.

Assume that there exists $x \in B \setminus C$. Removing $x$ from $B$ decreases
the denominators of the second term and the third terms. We need to show that
\[
	\frac{\abs{B \cap C}}{\abs{B \cup C} - 1} - \frac{\abs{A \cap B}}{\abs{A \cup B} - 1} \geq \frac{\abs{B \cap C}}{\abs{B \cup C}} - \frac{\abs{A \cap B}}{\abs{A \cup B}}\quad.
\]
Replacing $B \cup C$ with $A \cup B$ we have
\[
	\frac{\abs{B \cap C}}{\abs{A \cup B} - 1} - \frac{\abs{A \cap B}}{\abs{A \cup B} - 1} \geq \frac{\abs{B \cap C}}{\abs{A \cup B}} - \frac{\abs{A \cap B}}{\abs{A \cup B}}\quad.
\]
Multiplying both sides with $\abs{A \cup B}(\abs{A \cup B} - 1)$ leads us to
\[
	\abs{B \cap C}\abs{A \cup B} - \abs{A \cap B}\abs{A \cup B} \geq \abs{B \cap C}(\abs{A \cup B} - 1) - \abs{A \cap B}(\abs{A \cup B} - 1)\quad.
\]
Eliminating the terms leads us to inequality $\abs{B \cap C} \geq \abs{A \cap B}$
which is true because $A \cap B \subseteq C$. Thus removing $x$ from $B$ can only
increase the left side, hence we can safely assume that $B \subseteq C$, that is
$C = A \cup B$. This implies that $A \cup C = B \cup C = A \cup B$, $B \cap C = C$,
and $B \cap C = B$.  But then
\[
	\frac{\abs{A \cap C}}{\abs{A \cup C}} + \frac{\abs{B \cap C}}{\abs{B \cup C}} - \frac{\abs{A \cap B}}{\abs{A \cup B}} 
	 = \frac{\abs{A}}{\abs{A \cup B}} + \frac{\abs{B}}{\abs{A \cup B}} - \frac{\abs{A \cap B}}{\abs{A \cup B}}  = 1,
\]
which proves the theorem.
\qed
\end{proof}

\begin{proof}[of Theorem~\ref{thr:add}]
Let $\ifam{S} = \ifam{T} \cup \set{U}$. Corollary~\ref{cor:diff} implies that
\[
	\kl{\ifam{M}}{\ifam{T} \cup \ifam{B}} = \ent{\ifam{T} \cup \ifam{B}} - \ent{\ifam{M}} \geq \ent{\ifam{S} \cup \ifam{B}} - \ent{\ifam{M}} = \kl{\ifam{M}}{\ifam{S} \cup \ifam{B}}\quad.
\]
Thus
\[
\begin{split}
	\dist{\ifam{S}, \ifam{U} ; \ifam{B}} &= \frac{\kl{\ifam{M}}{\ifam{U} \cup \ifam{B}} + \kl{\ifam{M}}{\ifam{S} \cup \ifam{B}}}{\kl{\ifam{M}}{\ifam{B}}}  \\
	                                     &\leq \frac{\kl{\ifam{M}}{\ifam{U} \cup \ifam{B}} + \kl{\ifam{M}}{\ifam{T} \cup \ifam{B}}}{\kl{\ifam{M}}{\ifam{B}}}  = \dist{\ifam{T}, \ifam{U} ; \ifam{B}},
\end{split}
\]
which proves the result.
\qed
\end{proof}

\begin{proof}[of Theorem~\ref{thr:subset}]
Since $\ifam{M} = \ifam{T} \cup \ifam{B}$, the left side of the equation follows immediately.
Write $\ifam{S} = \ifam{U} \cup \ifam{B}$.  To prove the right side apply Corollary~\ref{cor:diff} to obtain
\[
	\frac{\kl{\ifam{M}}{\ifam{S}}}{\kl{\ifam{M}}{\ifam{B}}} = \frac{\ent{\ifam{S}}  - \ent{\ifam{M}}}{\kl{\ifam{M}}{\ifam{B}}}
	                                                        = \frac{\ent{\ifam{S}} - \ent{\ifam{B}} + \ent{\ifam{B}} - \ent{\ifam{M}}}{\kl{\ifam{M}}{\ifam{B}}}
	                                                        = 1 - \frac{\kl{\ifam{S}}{\ifam{B}}}{\kl{\ifam{M}}{\ifam{B}}},
\]
which proves the result.
\qed
\end{proof}

\begin{proof}[of Theorem~\ref{thr:np}]
We prove the hardness by reducing \textsc{ExactCover}, testing whether a cover
has a mutually disjoint cover, to the decision version of \textsc{Redescribe}.

Assume that we are given a ground set $S$ with $N$ elements and a cover
$\ifam{C} = \enset{C_1}{C_L}$ for $S$. To define the tiles and the target
frequencies we will first define a particular dataset $D$.  This dataset has $N + 1
+ L$ transactions of length $N$, an item representing an
item in the ground set $S$. The first $L$ transactions correspond to the cover
sets, $(i, j)$ entry is equal to $1$ if and only if $C_i$ contains the $j$th item.
The last $N + 1$ transactions are full of ones.

The tile sets are as follows. $\ifam{B} = \emptyset$. $\ifam{T}$ contains one tile
with the last $N + 1$ transactions and all items. $\ifam{U} = \enset{U_1}{U_L}$
contains $L$ tiles such that $\attr{U_i} = C_i$ and 
$\trans{U_i} = \set{i, L + 1, \ldots, L + N + 1}$. Let $\ifam{V}$ be a subset of $\ifam{U}$.
Write $f(\ifam{V}) = \sum_{V \in \ifam{V}}\abs{\attr{V}}$.
Theorem~\ref{thr:distanceexact} now implies that
\[
	\dist{\ifam{V}, \ifam{T}; \ifam{B}} = \frac{(N + 1)(N - \abs{\attr{\ifam{V}}}) + f(\ifam{V})}{(N + 1)N + f(\ifam{V})} \quad .
\]
If $\abs{\attr{\ifam{V}}} < N$, then we can augment with $V$ with a tile containing an item
outside of $\attr{\ifam{V}}$. This will reduce the first term in the numerator at least by $N + 1$
and increase the second term by $N$, at maximum. Thus such augmentation will decrease the distance
and so we can safely assume that $\abs{\attr{\ifam{V}}} = N$. This implies that
$\dist{\ifam{V}, \ifam{T}; \ifam{B}} = \frac{f(\ifam{V})}{(N + 1)N + f(\ifam{V})}$.
This is a monotonic function of $f(\ifam{V})$. Since $\attr{\ifam{V}}$ contains all items,
$f(\ifam{V}) \geq N$ and is equal to $N$ if and only if the tiles in $\ifam{V}$ are disjoint.
Hence \textsc{ExactCover} has a solution if and only if, there is a redescription with distance
equal to $ N / (N(N + 1) + N)$.
\qed
\end{proof}

\end{document}